\pdfoutput=1

\RequirePackage{luatex85}

\documentclass[a4paper,11pt]{article}
\usepackage{iftex}

\ifluatex
\else
\ifxetex
\else
\usepackage[utf8]{inputenc}
\usepackage[T1]{fontenc}
\fi
\fi

\usepackage{amsmath}
\usepackage{amssymb}
\usepackage{amsthm}
\usepackage{mathtools}
\usepackage{thmtools}

\makeatletter
\def\moverlay{\mathpalette\mov@rlay}
\def\mov@rlay#1#2{\leavevmode\vtop{%
   \baselineskip\z@skip \lineskiplimit-\maxdimen
   \ialign{\hfil$\m@th#1##$\hfil\cr#2\crcr}}}
\newcommand{\charfusion}[3][\mathord]{
    #1{\ifx#1\mathop\vphantom{#2}\fi
        \mathpalette\mov@rlay{#2\cr#3}
      }
    \ifx#1\mathop\expandafter\displaylimits\fi}
\makeatother

\newcommand{\cupdot}{\charfusion[\mathbin]{\cup}{\cdot}}

\ifluatex
\usepackage{fontspec}
\usepackage{unicode-math}
\else
\ifxetex
\usepackage{fontspec}
\usepackage{unicode-math}
\else
\usepackage{isomath}
\usepackage{libertine}
\usepackage[libertine]{newtxmath} % must load after ams packages
\usepackage[scaled=0.96]{zi4} % use Inconsolata as the typewriter font
\fi
\fi

\usepackage[DIV=11]{typearea} % large value = less margin, recommended values: 10pt: 8, 11pt: 10, 12pt: 12

\usepackage{microtype}

\usepackage[font=small]{caption}
\usepackage[algo2e,ruled,vlined,linesnumbered]{algorithm2e}

\usepackage{xcolor}

\usepackage[textwidth=2cm,textsize=footnotesize]{todonotes}
\usepackage{lineno}

\usepackage[
	style=alphabetic,
    maxalphanames=4,
    minalphanames=4,
    maxnames=20,
    minnames=20,
	backref=true,
	doi=true,
	url=true,
	backend=biber
]{biblatex}

\ifpdf
\usepackage[ocgcolorlinks]{hyperref} % Option ocgcolorlinks makes links only colored on screen and not on printouts
\else
\usepackage[hidelinks]{hyperref}
\fi
\usepackage{cleveref}

\allowdisplaybreaks[1]

\makeatletter
\g@addto@macro\bfseries{\boldmath}
\makeatother

\makeatletter
\g@addto@macro\mdseries{\unboldmath}
\g@addto@macro\normalfont{\unboldmath}
\g@addto@macro\rmfamily{\unboldmath}
\g@addto@macro\upshape{\unboldmath}
\makeatother

\makeatletter
\g@addto@macro\bfseries{\boldmath}
\def\thmhead@plain#1#2#3{%
  \thmname{#1}\thmnumber{\@ifnotempty{#1}{ }\@upn{#2}}%
  \thmnote{ {\the\thm@notefont\unboldmath(#3)}}}
\let\thmhead\thmhead@plain
\makeatother

\DeclareCiteCommand{\citem}
    {}
    {\mkbibbrackets{\bibhyperref{\usebibmacro{postnote}}}}
    {\multicitedelim}
    {}

\renewcommand*{\multicitedelim}{\addcomma\space}

\AtEveryCitekey{\clearfield{note}}

\makeatletter
\AtBeginDocument{%
    \newlength{\temp@x}%
    \newlength{\temp@y}%
    \newlength{\temp@w}%
    \newlength{\temp@h}%
    \def\my@coords#1#2#3#4{%
      \setlength{\temp@x}{#1}%
      \setlength{\temp@y}{#2}%
      \setlength{\temp@w}{#3}%
      \setlength{\temp@h}{#4}%
      \adjustlengths{}%
      \my@pdfliteral{\strip@pt\temp@x\space\strip@pt\temp@y\space\strip@pt\temp@w\space\strip@pt\temp@h\space re}}%
    \ifpdf
      \typeout{In PDF mode}%
      \def\my@pdfliteral#1{\pdfliteral page{#1}}% I don't know why % this command...
      \def\adjustlengths{}%
    \fi
    \ifxetex
      \def\my@pdfliteral #1{}% isn't equivalent to this one
      \def\adjustlengths{\setlength{\temp@h}{-\temp@h}\addtolength{\temp@y}{1in}\addtolength{\temp@x}{-1in}}%
    \fi%
    \def\Hy@colorlink#1{%
      \begingroup
        \ifHy@ocgcolorlinks
          \def\Hy@ocgcolor{#1}%
          \my@pdfliteral{q}%
          \my@pdfliteral{7 Tr}% Set text mode to clipping-only
        \else
          \HyColor@UseColor#1%
        \fi
    }%
    \def\Hy@endcolorlink{%
      \ifHy@ocgcolorlinks%
        \my@pdfliteral{/OC/OCPrint BDC}%
        \my@coords{0pt}{0pt}{\pdfpagewidth}{\pdfpageheight}%
        \my@pdfliteral{F}% Fill clipping path (the url's text) with
        \my@pdfliteral{EMC/OC/OCView BDC}%
        \begingroup%
          \expandafter\HyColor@UseColor\Hy@ocgcolor%
          \my@coords{0pt}{0pt}{\pdfpagewidth}{\pdfpageheight}%
          \my@pdfliteral{F}% Fill clipping path (the url's text)
        \endgroup%
        \my@pdfliteral{EMC}%
        \my@pdfliteral{0 Tr}% Reset text to normal mode
        \my@pdfliteral{Q}%
      \fi
      \endgroup
    }%
}
\makeatother

\usepackage{thm-restate}

\colorlet{DarkRed}{red!50!black}
\colorlet{DarkGreen}{green!50!black}
\colorlet{DarkBlue}{blue!50!black}

\hypersetup{
	linkcolor = DarkRed,
	citecolor = DarkGreen,
	urlcolor = DarkBlue,
	bookmarksnumbered = true,
	linktocpage = true
}

\DontPrintSemicolon
\SetKwProg{Procedure}{Procedure}{}{}
\SetProcNameSty{textsc}
\SetFuncSty{textsc}
\SetKw{KwAnd}{and}
\SetKw{KwDo}{do}

\SetCommentSty{mycommfont}

\declaretheorem[numberwithin=section]{theorem}
\declaretheorem[numberlike=theorem]{lemma}

\declaretheorem[numberlike=theorem]{definition}
\declaretheorem[numberlike=theorem]{claim}

\declaretheorem[numberlike=theorem, style=remark]{remark}

\DeclareMathOperator{\poly}{poly}
\DeclareMathOperator{\Span}{span}
\DeclareMathOperator{\rk}{rk}
\newcommand{\T}{\mathcal T}
\newcommand{\B}{\mathcal B}
\newcommand{\M}{\mathcal M}
\newcommand{\I}{\mathcal I}
\renewcommand{\P}{\mathcal P}
\newcommand{\eps}{\varepsilon}
\renewcommand{\epsilon}{\varepsilon}

\title{Dynamic Matroids: Base Packing and Covering}
\author{
Tijn de Vos\thanks{TU Graz, Austria. This research was funded in whole or in part by the Austrian Science Fund (FWF) \url{https://doi.org/10.55776/P36280}. For open access purposes, the author has applied a CC BY public copyright license to any author-accepted manuscript version arising from this submission.}
\and \textcircled{r}\thanks{The author ordering was randomized using \url{https://www.aeaweb.org/journals/policies/random-author-order/} generator. It is requested that citations of this work list the authors separated by \texttt{\textbackslash textcircled\{r\}} instead of commas.} \and
Mara Grilnberger\thanks{Department of Computer Science, University of Salzburg, Austria. This research has been supported by the EXDIGIT (Excellence in Digital Sciences and Interdisciplinary Technologies) project, funded by Land Salzburg under grant number 20204-WISS/263/6-6022. This project has received funding from the European Research Council (ERC) under the European Union's Horizon 2020 research and innovation programme (grant agreement No 947702).} 
}

\date{}

\begin{document}
\begin{titlepage}
\maketitle
\begin{abstract}
In this paper, we consider dynamic matroids, where elements can be inserted to or deleted from the ground set over time. The independent sets change to reflect the current ground set. As matroids are central to the study of many combinatorial optimization problems, it is a natural next step to also consider them in a dynamic setting. The study of dynamic matroids has the potential to generalize several dynamic graph problems, including, but not limited to, arboricity and maximum bipartite matching. We contribute by providing efficient algorithms for some fundamental matroid questions. 

In particular, we study the most basic question of maintaining a base dynamically, providing an essential building block for future algorithms. We further utilize this result and consider the elementary problems of base packing and base covering. We provide a deterministic algorithm that maintains a $(1\pm \varepsilon)$-approximation of the base packing number $\Phi$ in $O(\Phi \cdot \text{poly}(\log n, \varepsilon^{-1}))$ queries per update. Similarly, we provide a deterministic algorithm that maintains a $(1\pm \varepsilon)$-approximation of the base covering number $\beta$ in $O(\beta \cdot \text{poly}(\log n, \varepsilon^{-1}))$ queries per update. Moreover, we give an algorithm that maintains a $(1\pm \varepsilon)$-approximation of the base covering number $\beta$ in $O(\text{poly}(\log n, \varepsilon^{-1}))$ queries per update against an oblivious adversary.

These results are obtained by exploring the relationship between \emph{base collections}, a generalization of tree-packings, and base packing and covering respectively. We provide structural theorems to formalize these connections, and show how they lead to simple dynamic algorithms. 
\end{abstract}

\vfill

\textbf{Acknowledgements.} The authors would like to thank Aleksander Christiansen for the preliminary discussions that were the inspiration for this research.

\thispagestyle{empty}

\newpage
\thispagestyle{empty}
\tableofcontents

 \newpage
% \listoftodos

\end{titlepage}
\newpage

\section{Introduction}
Matroids generalize different mathematical concepts such as graphs and vector spaces. They have applications in combinatorial optimization, geometry, topology, network theory, and coding theory~\cite{oxley92,HlinenyW06,welsh2010matroid,recski2013matroid,Fujita24}. In particular, matroid problems are often seen as ``the problems where greedy algorithms are effective" (see, e.g.,~\cite{schrijver2003combinatorial}). Although there are dynamic algorithms for specific matroids, the work on dynamic algorithms for general matroids is limited. Concurrent work by Chandrasekaran, Chekuri, and Zhu~\cite{chandrasekaran2025online} also initiated the somewhat related study of \emph{online} matroids, where the matroid is slowly revealed over time. With this paper, we would like to develop the study of dynamic matroids by providing efficient algorithms for some fundamental matroid questions. In particular, we study the most basic question of maintaining a base dynamically, providing essential building blocks for future algorithms. Let us start with some definitions.

\paragraph{Matroids.}
Formally, a matroid is defined as a tuple $\M = (E, \I)$, where $E$ is a finite \emph{ground set} of elements and $\I\subseteq \P(E)$ a family of \emph{independent sets}, such that the following three properties hold: 
 1) Non trivial: $\emptyset \in \I$.
    2) Downward closure: if $A\in \I$ and $A' \subseteq A$, then $A' \in \I$. 
    3) Exchange property: if $A,B\in \I$ and $|A|> |B|$, then there is an element $e\in A\setminus B$ such that $B\cup \{e\} \in \I$.

We denote $n:=|E|$ to be the size of the ground set. The \emph{rank} of a set $A\subseteq E$ is defined as the size of the largest independent set it contains:

$
    \rk(A) := \max\limits_{\substack{A'\in \I \text{ s.t.}\\ A' \subseteq A}} |A'|$.

We say that $B\subseteq E$ is a \emph{base} of $\M$ if it is a maximal independent set, i.e., $\rk(B)=|B|=\rk(E)$. Given a weight function on the elements, a \emph{minimum weight base} is a base of the smallest total weight.  
Since $\I$ can be as large as $2^{|E|}$, it is often not given explicitly, but implicitly via oracle access. In this paper, we use a \emph{rank oracle}, which provides $\rk(A)$ upon a query $A\subseteq E$.

\paragraph{Matroid Problems.}
Two classic matroid problems are matroid union and matroid intersection. They reduce to each other in polynomial time, see, e.g., \cite{edmonds1970submodular,lawler1970optimal}.  These problems generalize many graph problems, such as packing disjoint spanning trees, computing the arboricity, bipartite matching, see, e.g., \cite{schrijver2003combinatorial}.
Base packing and base covering can be seen as instances of matroid union. In this paper, we investigate these two fundamental cases in the dynamic setting. 

One well-studied example of a matroid is the \emph{graphic matroid}, where $E$ is some set of edges and $A\subseteq E$ is independent if and only if it is acyclic. Base packing corresponds to packing disjoint spanning trees and the base covering number corresponds to the arboricity of the graph. Throughout, we generalize results for the graphic matroid and we will state the algorithms for graphic matroids as a comparison. 

Algorithms for base packing and covering have been studied for a long time, see, e.g.,~\cite{knuth1973matroid,Cunningham86,Karger93,Karger98,ChekuriQ17,Blikstad21}. 
The state of the art is given by
Quanrud~\cite{Quanrud24}. They provide exact algorithms that use $\tilde O(n+k\cdot \rk(E)^2)$ independence-queries\footnote{For simplicity, we use the notation $\tilde O(f):=O(f\poly\log f)$.}, where $k$ is the packing/covering number respectively. They also provide $(1+\eps)$-approximations in $\tilde O(n/\eps)$ independence-queries.

\paragraph{Dynamic Matroids.}
For dynamic matroids, we consider element insertions and deletions to the ground set.\footnote{The \emph{decremental} version (deletions only) of this has been studied in~\cite{BlikstadMNT23}, and we argue below that this \emph{fully dynamic} version is the natural and most general definition.} 
When an element $e$ is deleted, the independent sets $\I$ simply restrict to the sets without $e$: $\{I\in \I : e\notin I\}$. In other words, $\M$ is restricted to $E\setminus \{e\}$. When an element $e$ is inserted, the adversary also decides on a collection $\I_e$ such that $(E\cup\{e\}, \I \cup \I_e)$ is a matroid and $e\in I$ for every $I\in \I_e$. This means that inserting and then deleting the same element results in the same matroid. The algorithm receives the element updates and can query the (new) independence sets. 

%Note that in dynamic graph algorithms, queries have a different role. For example, when solving all-pairs shortest path problems, it might not be feasible to maintain the solutions explicitly, since writing each entry of the solution can be the bottleneck, rather than the computation. In such situations a distance oracle is maintained. The query time then corresponds to the time needed to answer the question. In this work, we maintain all answers explicitly. Hence our queries always refer to the matroid oracle. 

An alternative definition of dynamic matroids could be as follows. One could allow for updates to $\mathcal{I}$ that do not stem from an element deletion or insertion. However, if we can completely change $\mathcal{I}$ in a single update, then we could go from any matroid on $n$ elements to any other matroid on $n$ elements. This means that any lower bound for a static algorithm carries over to \emph{each} update of the dynamic algorithm. Hence, recomputing from scratch is the best one could do. On the other extreme, we could only allow adding or deleting \emph{a single set} $I$ from $\I$. We note that, in general, such a change to $\I$ will no longer guarantee that it is a matroid: consider adding $I$ to $\I$. By the downward closure property, all subsets of $I$ also have to be in $\I$. We consider the most restrictive (and hence most general) version of this, where we only add $I$ if $I\setminus \{e\}$ is already in $\I$ for some element $e$. The alternative view is that the element $e$ is added to the ground set\footnote{If $e$ was already part of another independent set $e\in I' \in\I$, we can model this by deleting and inserting~$e$.}, and the independent sets are updated accordingly. 

We distinguish two types of adversaries: an \emph{oblivious adversary} fixes the updates beforehand, independent of the random choices in the algorithm, while an \emph{adaptive adversary} determines the next update depending on the current state of the algorithm. In this paper, our algorithms are either deterministic, hence hold against an adaptive adversary, or are randomized and hold against an oblivious adversary.

We note that dynamic matroids have been studied in the realm of submodular function maximization over dynamic matroids, see, e.g., \cite{MirzasoleimanK017,ChenP22,BanihashemBGHJM24}.

\paragraph{Scope.}
We approach this model by first investigating the classical matroid problem of a minimum weight base and then considering the fundamental problems of base packing and covering. Although the former has been studied in~\cite{BlikstadMNT23} (see the discussion below \Cref{lm:dyn_min_base}), there are no results regarding the latter two. However, in the special case of the graphic matroid, all three problems have been well studied.

%Dynamic greedy algorithms have also been studied in the case of spanners. There is a folklore greedy algorithm that gives an optimal spanner. Bhattacharya, Saranurak, and Sukprasert~\cite{BhattacharyaSS22} show that this can lead to dynamic algorithms with low recourse. However, the update time for this greedy algorithm is polynomial. This shows that greedy algorithms do not always interact well with dynamic updates. For matroids, where static greedy algorithms are very successful, we attempt to find the approaches that allow for efficient dynamic updates. 

\paragraph{Minimum Weight Base.}
First, we give our result for maintaining a minimum weight base. 
\begin{restatable}{proposition}{DynMinBase}\label{lm:dyn_min_base}
    There exists a deterministic algorithm that, given a dynamic matroid $\M$ with weight function $w\colon E \to [1, \dots, W]$, maintains a minimum weight base, where each update uses $O(\log n)$ rank-queries. 
\end{restatable}

Blikstad, Mukhopadhyay, Nanongkai, and Tu~\cite{BlikstadMNT23} provide an algorithm for maintaining a minimum weight basis under deletions. Each update requires $\tilde O(\sqrt{\rk(E)})$ worst-case rank-queries. However, using a `dynamic rank oracle', they also ensure $\tilde O(\sqrt{\rk(E)})$ worst-case update \emph{time}, as it can only answer queries, where the answer can be computed efficiently on a concrete matroid. The algorithm is based on the MST algorithm with $\tilde O(\sqrt{|V|})$ update time by Frederickson~\cite{Frederickson85}, combined with the sparsification technique of Eppstein, Galil, Italiano, and Nissenzweig~\cite{EppsteinGIN97}. 
It seems likely that such a result can be extended to the fully dynamic setting. Blikstad et al.~\cite{BlikstadMNT23} use this as a subroutine and the authors seem to have optimized their result for their application in solving matroid union. 
The difference to our result, \Cref{lm:dyn_min_base}, is that we have a much lower query time, but do not give guarantees on the update time, since the computation time for answering a query depends on the concrete matroid.

Both the algorithm~\cite{BlikstadMNT23} and our algorithm from \Cref{{lm:dyn_min_base}} use a rank oracle. 
Sometimes, algorithms are developed using an \emph{independence oracle}, which only provides whether $A\in \I$. Clearly, the rank oracle is stronger than the independence oracle. Statically, it is even known that it is \emph{strictly} stronger: computing a minimum weight base can be done with $n$ simultaneous rank-queries, while this needs $\tilde \Omega(n^ {1/3})$ rounds of simultaneous independence-queries~\cite{KarpUW88}. 
It would be interesting to see if there exists a decremental algorithm for maintaining a minimum weight base using a sublinear number of independence queries. 

\paragraph{Minimum Weight Base in the Graphic Matroid.}
The dynamic minimum spanning tree (MST) problem is one of the most studied problems in dynamic graph algorithms, see, e.g.,~\cite{Frederickson85, EppsteinITTWY92,EppsteinGIN97,Frederickson97,AlbertsH98,HenzingerK01,HolmLT01, NanongkaiSW17,Wulff-Nilsen17}. The special case of an unweighted graph is the dynamic spanning tree problem (see, e.g.,~\cite{ HenzingerT97,HenzingerK99,Thorup00,PatrascuT07,KapronKM13,GibbKKT15,Wulff-Nilsen16a, NanongkaiS17,HuangHKPT23}), which also has close ties to dynamic connectivity. 
Let us highlight the state of the art: Holm, de Lichtenberg, and Thorup~\cite{HolmLT01} maintain an MST deterministically with
$O(\log^4 |E|)$ amortized update time. Nanongkai, Saranurak, and Wulff-Nilsen~\cite{NanongkaiSW17}~provide a Las Vegas algorithm with $|V|^{o(1)}$ worst-case update time.

\paragraph{Packing and Covering.}
The other two problems we treat in this paper are \emph{base packing} and \emph{base covering}. In \emph{base packing} the goal is to pack as many disjoint bases in $\M$ as possible. Formally, we define the \emph{(fractional) packing number} $\Phi_\M$ of a matroid $\M$ as follows\begin{equation*}\Phi_\M := \min_{\substack{A\subseteq E\ \rm{s.t.} \\ \rk(\overline A) < \rk(E)}}\frac{|A| }{\rk(E)-\rk(\overline A)}.
\end{equation*}

This is also known as matroid strength. The integral packing number is $\lfloor \Phi_\M \rfloor$. Edmonds~\cite{Edmonds65} showed that $\lfloor \Phi_\M \rfloor$ equals the number of disjoint bases that can be packed in $\M$. 

Dual to base packing, in \emph{base covering} the goal is to cover $\M$ by as few bases as possible.
Formally, we define the \emph{(fractional) covering number} $\beta_\M$ of a matroid $\M$ as follows
\begin{equation*}\beta_\M := \max_{\substack{A\subseteq E\ \rm{s.t.} \\  A \neq \emptyset}}\frac{|A| }{\rk(A)}.
\end{equation*}
This is also known as matroid density. We define the integral covering number as $\lceil \beta_\M\rceil$. Edmonds~\cite{Edmonds65partition} showed that $\lceil \beta_\M\rceil$ equals the minimum number of bases necessary to cover $\M$.
We omit the subscript if the matroid is clear from the context.

Computing exact packing and covering is a hard question. Even in certain specific matroids, like the graphic matroid (see the paragraph below) this is relatively slow. 
In this paper, we aim for efficient algorithms with $\poly\log n$ queries per update. In particular, we give the first dynamic algorithms for approximating the fractional packing number and the fractional covering number.

\begin{restatable}{theorem}{DynPackDet}\label{thm:dyn_pack_det}
    Let $\eps\in (0,1/2)$ be a parameter and let $\M$ be a dynamic matroid that at any moment contains at most $n$~elements. If $\Phi$ is upper bounded by $\Phi_{\max}$, we can deterministically maintain a $(1\pm \eps)$-approximation of the fractional packing number with $O(\Phi_{\max}^2\cdot \eps^{-4} \cdot\log^3 n)$ worst-case rank-queries per update or $O(\Phi_{\max}\cdot \eps^{-4} \cdot\log^3 n)$ amortized rank-queries per update.
\end{restatable}

Since we give the first dynamic algorithm, our result can only be compared to using deterministic static algorithms to recompute the fractional packing number from scratch after every update. Using the state of the art by Chekuri and Quanrud~\cite{ChekuriQ17}, we get $\tilde{O}(n\Phi_{\max} / \epsilon^2)$ queries per update.\footnote{Chekuri and Quanrud~\cite{ChekuriQ17} only require the weaker independence-queries. However, even using a rank-oracle there is no known algorithm with $o(n)$ queries. The same holds for the results by \cite{ChekuriQ17,Quanrud24} stated below.}  Hence, our algorithm provides an exponential improvement in terms of $n$.

\begin{restatable}{theorem}{DynCovDet}\label{thm:dyn_cov_det}
    Let $\eps\in (0,1/2)$ be a parameter and let $\M$ be a dynamic matroid that at any moment contains at most $n$~elements. If $\beta$ is upper bounded by $\beta_{\max}$, we can deterministically maintain a $(1\pm \eps)$-approximation of the fractional covering number with $O(\beta_{\max}^2\cdot \eps^{-4} \cdot\log^3 n)$ worst-case rank-queries per update. 
\end{restatable}

Again, there are no preexisting dynamic algorithms to compare our result to.
Even static deterministic base covering is not as well studied. Quanrud~\cite{Quanrud24} conjectures that the techniques from Chekuri and Quanrud~\cite{ChekuriQ17} extend to deterministic approximate base covering in $\tilde O(n \beta_{\max}/\eps^2)$ queries, which would transfer to $\tilde O(n \beta_{\max}/\eps^2)$ queries per update. 

Using randomness, we obtain an algorithm independent of the covering number against an oblivious adversary. 
\begin{restatable}{theorem}{DynCov}\label{thm:dyn_cov}
    Let $\eps\in (0,1/2)$ be a parameter and let $\M$ be a dynamic matroid that at any moment contains at most $n$~elements.  There exists a fully dynamic algorithm that maintains a $(1\pm \eps)$-approximation of the fractional covering number with $O(\log^6 n/\eps^8)$ worst-case rank-queries per update. The algorithm is correct with high probability against an oblivious adversary.  
\end{restatable}

Once again, we compare to recomputing from scratch after every update using a static algorithm and obtain an exponential improvement. The state of the art for computing the static fractional covering number is by Quanrud~\cite{Quanrud24} and would result in $\tilde O(n + \rk(E) / \epsilon^3)$ queries per update. The same technique cannot be applied in the packing case to remove the dependence on $\Phi_\text{max}$, as we discuss in \Cref{sc:dynCov}.

\paragraph{Packing and Covering in the Graphic Matroid.}
Dynamic tree-packing has been implicitly studied by Thorup~\cite{Thorup07}. This paper uses dynamic tree-packing due to its relation to min-cut. However, it implies a $(1\pm\eps)$-approximation of the fractional tree packing number\footnote{We realize that the naming conventions here overlap. Rather than disregarding the conventions completely, we write `tree-packing' for the collection of trees and `tree packing' when talking about the tree packing number.} $\Phi\le \Phi_{\max}$ in $\tilde O(\Phi_{\max}^2 \log^6|E|/\eps^{-4})$ amortized update time. 

De Vos and Christiansen~\cite{deVosC24} give fully dynamic algorithms for $(1\pm\eps)$-approximate arboricity with $O(\poly(\log |E|, \eps^{-1}))$ update time against an adaptive adversary. 
Banerjee, Raman, and Saurabh~\cite{BanerjeeRS20} maintain the exact arboricity with $\tilde O(|E|)$ update time. 
As such, we believe that an algorithm with $\poly \log(|E|)$ queries per update for \emph{exact} packing/covering for general matroids requires a breakthrough in techniques. We focus on obtaining $(1\pm \eps)$-approximate results with $\poly \log(|E|)$ queries per update.

\paragraph{Other Applications.}
The graphic matroid is one of the most studied matroids, where the above references show that dynamic base packing and covering has been studied directly. In other matroids, this is not the case. However, base packing and covering (and hence their dynamic versions) have other interesting applications. Here, we mention two: the linear matroid and the partition matroid, see e.g.~\cite{oxley92} for the definitions and connections. 

In the linear matroid, base packing corresponds to the ability to decompose a matrix into many full-rank disjoint column sets. Base covering corresponds to finding multiple sets that form a basis for the spanned vector space, which has applications in (network) coding, see, e.g.,~\cite{dougherty2007networks,casazza2013introduction}.
In the partition matroid and generalizations thereof, base packing and covering correspond to multi-way assignment and various types of scheduling, see, e.g., \cite{burkard1990constrained,kawase2021optimal}. 

Moreover, base packing and covering are highly connected to Shannon switching games~\cite{shannon1948mathematical,lehman1964solution}.

\subsection{Technical Overview}
Our main technical tool are \emph{base collections}, which we will introduce first. Then, we will sketch how we use them to obtain our packing and covering results, \Cref{thm:dyn_pack_det,thm:dyn_cov_det,thm:dyn_cov}.

\subsubsection{Base Collections}\label{sec:overview_base_collactions}
An important tool for dynamic arboricity is tree-packing. This concept has first appeared in the seminal works by Nash-Williams~\cite{Nash61} and Tutte~\cite{Tutte61}. It has also been well studied in its relation to the minimum cut of the graph~\cite{gabow1995matroid, karger2000minimum, ThorupK00,Thorup07,Thorup08,ChekuriQ019, DBLP:conf/stoc/DagaHNS19,DoryEMN21,deVosC24}. 
We generalize the concept of tree-packings to matroids, and call it a \emph{base collection}, to avoid confusion with respect to packing disjoint bases.

A \emph{base collection} $\B$ is a family of bases $B$ for the matroid $\M$, allowing multiple occurrences of the same base. The \emph{load} of an element $e\in E$ is defined as $L^{\B}(e):=|\{B\in \B: e\in B\}|$. The \emph{relative load} is defined as  $\ell^{\B}(e)=L^{\B}(e)/|\B|$. Whenever the base collection is clear from context, we omit the superscript. 

Next, we want to define some `ideal' relative load. Hereto, we first prove a corollary using Edmond's theorem on base packing (\Cref{thm:matroids_edmonds}). 
We show that the lowest possible maximum load in base collections relate to the base packing number. 
\begin{restatable}{corollary}{matroidNW} \label{cor_phi} We have that $\max_{\mathcal B}\frac{1}{\max\limits_{e\in E} \ell^\B(e)} = \Phi_{\mathcal{M}}$.
\end{restatable}

However, for purpose of analysis, we would like a \emph{specific} packing with this property. More precisely, we would like the \emph{loads} corresponding to that packing. This we call the \emph{ideal relative loads}, a generalization of the `ideal tree-packing' for the graphic matroid~\cite{Thorup07}. 

\paragraph{Ideal Relative Loads.}
We define the concept of ideal loads, which are hard to compute, but capture the structural properties of the graph well. 
First, we define the \emph{restricted matroid} $\M|A$ as $\M|A := (A, \mathcal{I}\;|\;A)$, where $\mathcal{I}|A := \{ X \in \mathcal{I} \; | \; X \subseteq A\}$.
Then, we assign ideal relative loads $\ell^*(e)$ for all $e \in E$:\begin{itemize}
\item Let $A_0 \subset E$ be a set such that $\frac{|A_0| }{\rk(E)-\rk(\overline{A_0})} = \Phi_\mathcal{M}$.
\item For all $e \in A_0$, set $\ell^*(e) = 1/\Phi_\mathcal{M}$.
\item Recurse on the matroid $\mathcal{M}|\overline{A_0}$.
\end{itemize}
We show that this is well-defined, meaning that the resulting ideal relative load values are independent of the concrete choices of the sets $A_0 \subseteq E$ if there are multiple such sets that give $\frac{|A_0| }{\rk(E)-\rk(\overline{A_0})} = \Phi_\mathcal{M}$, in \Cref{sc:structural}. 

\paragraph{Greedy Base Collections.}
 In this paper, we consider base collections $\B = \{B_1, \ldots, B_k\}$ built greedily as follows: the $i$-th base $B_i$ is a minimum weight base where the weights for each element are given by the relative loads induced by the base collection up until this base $\{B_1, \ldots, B_{i-1}\}$.

In the graphic matroid, it has been shown~\cite{Thorup07} that a greedy tree-packing $\T$ of size $\Theta (\Phi \log |V|/ \eps^2)$ approximates the ideal packing well, resulting in every element having a relative load that differs from the ideal relative load by a small additive error as follows
\begin{equation}
    |\ell^\T(e) - \ell^*(e)| \leq \eps/\Phi \label{eq:thorup_approx}
\end{equation}
for all $e \in E$.
We show a stronger statement for general matroids, for which we introduce the parameter $\gamma$ to provide a better trade-off between the error and the number of greedy bases. 
The parameter $\gamma$  satisfies $\Phi \le \gamma \le \beta$, which implies that $\Phi \le \beta$ -- which we have not shown yet. When considering the integer equivalents, the intuition is that $\lfloor\Phi\rfloor$ is the number of disjoint bases that fit in the matroid, and $\lceil\beta\rceil $ is the number of bases needed to cover the matroid, hence clearly $\lfloor\Phi\rfloor\le\lceil\beta\rceil $. The fact that also $\phi \le \beta$ is somewhat more technical, but easy to see using the structural results of \Cref{sec:overview_structural}.
\begin{restatable}{lemma}{GreedyColNew}\label{lem:greedy_collnew}
    Let $\gamma \in [\Phi, \beta]$ and let $\B$ be a greedy base collection with $|\B|\geq  3 \gamma \log n /\eps^2$. Then 
    \begin{equation}\label{eq:greedy_coll_high}
        |\ell^{\B}(e) - \ell^*(e)| \leq \eps \ell^*(e)
    \end{equation}
    for all $e\in E$ with $\ell^*(e) \geq 1/\gamma$ and 
    \begin{equation}\label{eq:greedy_coll_low}
        |\ell^{\B}(e) - \ell^*(e)| \leq \eps/\gamma
    \end{equation}
    for all $e \in E$ with $\ell^*(e) \leq 1/\gamma$.
\end{restatable}

The improvement over \cite{Thorup07} in the special case of the graphic matroid can be seen as follows: 
\Cref{eq:thorup_approx} can be reformulated, by picking $\epsilon$ accordingly, to obtain $ |\ell^{\B}(e) - \ell^*(e)| \leq \eps/\gamma $. However, it needs $|\B|=\Theta\left( \tfrac{\gamma^2}{\Phi}\log n/\eps^2\right)$ bases. 
That is, it has a \emph{quadratic} dependence on $\gamma$, where we obtain a \emph{linear} dependence. 

In other words, Thorup~\cite{Thorup07} showed that the values are an \emph{additive approximation}, since the error is independent of $\ell^*(e)$. We show a \emph{multiplicative approximation}, since the error depends linearly on $\ell^*(e)$. For our application, the latter leads to stronger results. To be precise, \Cref{lem:greedy_collnew} shows that a collection of at least $ 3 \beta \log n /\eps^2$ bases gives a $(1 \pm \eps)$-approximation of the covering number (using \Cref{thm:struct_covering}). A generalization of Thorup's version to matroids would give $|\B|=\Theta\left( \tfrac{\beta^2}{\Phi}\log n/\eps^2\right)$. 

The proof of \Cref{lem:greedy_collnew} utilizes a technique by Young \cite{young95}. We consider a distribution of bases, such that picking the bases for a base collection randomly from this distribution would result in the relative loads being ideal in expectation. Then we analyze the number of times the relative load
of an element does not approximate its ideal load well. We do this by replacing the randomly picked bases one after the other using a greedy approach to compute the new base. During this process, we consider pessimistic estimators for the number of violations at every step and show that they cannot be increased when a greedy base is added. This is also where the main difference to the graph case lies, as we need to consider the matroids restricted to elements or contracted to elements with certain ideal loads. The proof using the contraction is complicated by the fact that there is no notion of vertices in a general matroid. However, it still holds that any minimum weight base of $\M$ contains a minimum weight base for the matroid contracted to a subset.

To the best of our knowledge, the community was not aware of this property, given in \Cref{lem:greedy_collnew}, for the graphic matroid either. In particular, it means that we simplify the result of de Vos and Christiansen~\cite{deVosC24} for arboricity. They introduce an intricate procedure to artificially maintain $\Phi\approx \beta$ by adding virtual edges.  

\paragraph{Dynamic Data Structure for Base Collections.}
Similar to the case of the graphic matroid, we show that we can maintain a greedy base collection efficiently under dynamic updates. The main building block for this is maintaining a dynamic minimum weight base, \Cref{lm:dyn_min_base}. We maintain this by exploiting their greedy properties, combined with a binary search, to design a simple, efficient algorithm. For details, see \Cref{sc:greedy}. 

We then build a greedy base collection, by maintaining a minimum weight base for every base in the collection. We show how to handle the updates themselves and the recourse from changes in the weights due to the update. 
This gives the following result. 

\begin{restatable}{lemma}{GreedyBaseCol}\label{lm:dyn_greedy_base_col}
    There exists a deterministic algorithm that, given a matroid $\M$, maintains a greedy base collection $\B$ of size $|\B|$ with $O(|\B|^2\log(|\B|n))$ worst-case rank-queries per update.
\end{restatable}

\subsubsection{Structural Results}\label{sec:overview_structural}
Next, we consider how (ideal) base collections are related to packing and covering.

Note that, by the definition of $\ell^*$, the maximum load in base collections is related to the base packing number.
\begin{remark}\label{remark}
    We have that $\frac{1}{\max\limits_{e\in E}\ell^*(e)} = \Phi_\M$.
\end{remark}

We also show that the minimum load in base collections is related to the base covering number. This generalizes the result of de Vos and Christiansen~\cite{deVosC24}, who show the analogous result in the graphic matroid, where the covering number is called the arboricity. This insight is the base of the recent breakthrough by Cen et al.~\cite{CenFLLP25} that gives the first improvement on the running time of computing the arboricity in thirty years.

\begin{restatable}{theorem}{matroidPacking}\label{thm:struct_covering}
We have that
    $ \frac{1}{\min\limits_{e\in E}\ell^*(e)} = \max\limits_{\substack{A\subseteq E\ \rm{s.t.} \\ A \neq \emptyset}} \frac{|A|}{\rk(A)}$.
\end{restatable}
The proof follows the same lines as~\cite{deVosC24}, but has some more subtleties. To see this, we recall that the covering number is defined as 
\begin{equation*}
    \beta_\M := \max_{\substack{A\subseteq E\ \rm{s.t.} \\  A \neq \emptyset}}\frac{|A| }{\rk(A)}.
\end{equation*}

For the special case of graphic matroids, the arboricity of a graph $G=(V,E)$ is defined as $
 \max_{\substack{S \subseteq V \ \rm{s.t.} \\|S| >1 }}\frac{|E(S)| }{|S| - 1}$,
where $E(S)$ is the set of edges of $G[S]$.
The fact that the denominator of the equation changes from the size of a set to the rank of a set impacts the proof. The reason is that $|S \cupdot T|=|S| + |T|$ but we do not always have $\rk(A\cupdot B) \neq \rk(A) + \rk(B)$. 
We show that the \emph{submodularity} of the rank function (see, e.g, \cite{schrijver2003combinatorial}) suffices, i.e.,
\begin{equation*}
    \rk(A\cup B) + \rk(A \cap B) \leq \rk(A) + \rk(B).
\end{equation*}

\subsubsection{Dynamic Packing and Covering}
Using the results thus far, the worst-case result for dynamic base packing, \Cref{thm:dyn_pack_det}, follows quite immediately: using \Cref{remark}, we know that the packing number can be expressed in terms of the ideal relative loads. We can approximate this by a greedy base collection, see \Cref{lem:greedy_collnew} and for more details refer to the proof of \Cref{thm:dyn_pack_det} in \Cref{sc:dynPack}. And finally, we know how to maintain this under dynamic updates using \Cref{lm:dyn_greedy_base_col}.

The amortized result in \Cref{thm:dyn_pack_det} uses the fact that an element is not contained in every base of the collection, and hence a better recourse argument is possible. This follows the same lines of argument as the case of the graphic matroid~\cite{deVosC24} but has some nuances. We want to highlight a special case: an insertion can \emph{decrease} $\Phi$ (similarly, a deletion can \emph{increase} $\Phi$). Note that in graphs, this does not appear: such updates change the graph from disconnected to connected and the packing number of a disconnected graph is 0.
 In a matroid, the addition of this element needs to increase the overall rank, the packing number is not 0 and $e$ is part of \emph{any} base in $\M$ after the update. This breaks some of the argumentation for graphic matroids. In \Cref{sc:dynPack} we show how to handle such technicalities efficiently. 

 The deterministic result for dynamic base covering, \Cref{thm:dyn_cov_det}, is obtained in a similar manner. The number of worst-case rank-queries depends on the value known to upper bound the covering number at any point in time. To remove this dependency we make use of a sampling technique.

\paragraph{Sampling.} 
We use \emph{uniform sampling}, where every element is sampled with equal probability. This a standard technique in designing (graph) algorithms (see e.g.~\cite{McGregorTVV15}). The approach is similar to the sampling in multi-graphs when maintaining the arboricity~\cite{deVosC24}.
Although the sampling itself is the same as for multi-graphs -- it is uniform sampling over the elements -- the analysis is more involved for matroids.

Suppose we know (an approximation of) $\beta$. 
The idea is to sample with $p= \frac{\log n}{\eps^2 \beta}$, such that the covering number in the resulting sampled matroid will be $\Theta(\log n/\eps^2)$ with high probability. A constant approximation of $\beta$ suffices, which we can get from the approximation before the update. In total, we maintain $\log n$ copies of the algorithm: one for each possible estimate of $\beta$. 
The $\log n/\eps^2$ term here stems from the use of Chernoff bounds. 

Given this uniform sampling probability $p$, we consider the expressions $\frac{|S|}{\rk(S)}$ for each set $S$. 
Consider a set $S$ that maximizes this -- the case of sets with smaller values is omitted in this overview for brevity, see \Cref{thm:dyn_cov} for details. We show that their value will remain between $(1 - \eps) p \beta$ and $(1 + \eps) p \beta$ with high probability, where $\beta$ is the covering number in the original matroid. 
Obviously, the size of the set,~$|S|$, behaves accordingly. 
The hard part is showing that the rank of the set, $\rk(S)$, behaves well under sampling, i.e., does not change by more than a factor $1\pm \eps$ with our choice of $p$.
For graphic matroids, the expression simplifies by $\rk(S)=\#\text{vertices in }S$, which is clearly not affected by sampling edges. Hence these complications only occur for general matroids.

To investigate this, we first fix a rank $r$ and consider all sets $S$ of $\rk(S)=r$. When inspecting $\frac{|S|}{\rk(S)}$, we remark for our applications that we are only interested in the \emph{maximum} among such sets, which restricts the number of possible sets to $n^r$. This allows us to use a union bound to obtain our result. For more details, see \Cref{sc:dynPack} 

For the case of base packing a similar approach cannot be used. 
Again, the goal is to show how $\frac{|S|}{\rk(E)-\rk(\overline S)}$ behaves under uniform sampling. As before $|S|$ can be bounded by a Chernoff bound to be $(1\pm\eps)p|S|$. However, with the techniques known to us, it seems hard to show that $\rk(E)-\rk(\overline S)$ changes by at most a factor $(1\pm \eps)$ \emph{with high probability}. It remains an open question to obtain a dynamic base packing algorithm that updates independent of the packing number. 

% We would need to consider sets that would reduce the rank by a certain amount when removed from the matroid. The relevant sets $S$ can then have varying different ranks as well as different structures depending on the sampled elements. In particular, we do not have a property similar to $S$ having $S = \Span_\M(S)$ that would allow us to bound the number of relevant sets per rank nicely.

For graphic matroids, \cite{deVosC24} show a trick against \emph{adaptive adversaries}, where each vertex has ownership over some of the edges. Whenever one of its edges is affected by an update, it resamples all its edges. It turns out that this is strong enough to obtain results against adaptive adversaries. For general matroids, it is unclear how to bucket the elements in a way such that resampling one bucket every update protects against an adaptive adversary. It remains an open question how to design efficient algorithms in this case.

\subsubsection*{Independent Work}
The concurrent work of Arkhipov and Kolmogorov~\cite{arkhipov2025greedy} also obtains two of our results: \Cref{lem:greedy_collnew} for the special case that $\gamma =\beta$ and \Cref{thm:struct_covering}. They use these results to show that they can maintain an approximation to the density of a graph by packing pseudoforests. 

\subsection{Organization}
In \Cref{sc:greedy}, we show how to maintain a minimum weight base dynamically. 
In \Cref{sc:dynPC}, we show that dynamic base collections can be used to obtain the dynamic packing and covering results. 
Hereto, we show the structural results on base packing, base covering, and their connection to base collections in \Cref{sc:structural}. And finally, we show that greedy base collections approximate the ideal base collections and how to maintain them in \Cref{sc:greedyIdeal}. 

 \newpage
\section{Maintaining a Dynamic Minimum Weight Base}\label{sc:greedy}
In this section, we give an algorithm to maintain a minimum weight base.

\paragraph{Minimum Weight Base.}
Given a universe $E$, and weight function $w\colon E \to [1, \dots, W]$, a \emph{minimum weight base} is a base $B$ of minimum total weight $w(B)=\sum_{e\in B}w(e)$. Note that if the weights are unique, the minimum weight base is unique. For simplicity, we would like the weights to be unique without increasing the maximum weight. The minimum weight base of a matroid can be computed using a simple greedy approach as shown in~\cite{gale68, Edmonds71}. Here, the exact weights of the elements are never taken into account. It suffices to consider the order on the elements implied by their weights. Thus we can essentially re-weight the elements at each update, assigning each element a unique weight in $[1, \dots, n]$ respecting, firstly, the weights given by $w$, and secondly, in case of equal weight, the lexicographic order of the unique ids.\footnote{Note that after each update, many weights can change. However, this does not require additional queries, so does not affect our query complexity.}
For a \emph{dynamic minimum weight base}, the goal is to maintain a minimum weight base under element insertions and deletions. 
More generally, we can also allow for weight increases or decreases, but this can also be modeled by deleting the element and inserting it again with the new weight.

\paragraph{Preliminaries.}
The \emph{span} of a set $A \subseteq E$ is denoted by $\Span(A)$ and is defined as the set $\{ e \in E | \rk(A) = \rk(A + e)\}$. We say that the set $A$ spans element $e \in E$ or that $A$ spans a set $B \subseteq E$ if $e \in \Span(A)$ or $B \subseteq \Span(A)$, respectively. The set $A$ always spans itself as well as all other elements that can be added to $A$ without increasing the size of the maximal independent subset.

A \emph{circuit} $C\subseteq E$ is an inclusion-wise minimal dependent set of elements, i.e., $C \notin \I$ and for every $x\in C$, $C\setminus\{x\}\in \I$. The unique circuit in $S + e$ for $S \in \I$ and $ S + e \not \in \I$ is denoted by $C( S + e)$. We have the following lemma concerning circuits. 
\begin{lemma}[\cite{schrijver2003combinatorial}]\label{lm:circuits}
    Let $C$ and $C'$ be circuits. If $x\in C\cap C' $ and $y\in C\setminus C'$, then there exists a circuit in $(C\cup C') \setminus \{x\}$ containing $y$. 
\end{lemma}

\paragraph{Algorithm.}
We base our algorithm on two fundamental properties of minimum spanning trees (see e.g., \cite{kleinbergT06}), which also hold for minimum weight bases in matroids. 
\begin{itemize}
    \item The \emph{cycle property}: For any cycle, the edge with the largest weight cannot be in the MST.
    \item The \emph{cut property}: For any cut, the edge with the smallest weight that crosses the cut is in the MST.
\end{itemize}
In the proof below, we implicitly use and prove these properties for matroids. Note that in this context, a `cut' is a minimal set of elements that reduces the rank.

\DynMinBase*
\begin{proof}
Without loss of generality, we assume that the weights are unique. We can do tie-breaking of elements with equal weight by lexicographic order. 

We introduce the following notation: the sets $B_{\le t} := \{f \in B: w(f) \le t\}$ and similarly ${E_{\le t} := \{f \in E: w(f) \le t\}}$. Let us first consider element insertion. If the insertion increases the rank of the matroid, we can simply add the new element to $B$ to get a minimum weight base for the new matroid. Otherwise, suppose $e$ is inserted, finding the circuit in $B + e$ can be costly. However, we just need to find the element $f$ of highest weight on this circuit (possibly $e$ itself if it was inserted with a high weight). Formally, we do this as follows. 
\begin{itemize}
    \item Find min $t$ s.t.\ $\rk(B_{\leq t}) = \rk(B_{\leq t} + e)$.
    \item Output the unique element $f$ with weight $t$. 
\end{itemize} 

Note that now $f$ is the element of highest weight on the circuit of $B + e$. 
Next, we show that we can use this method to maintain a minimum base correctly over a sequence of insertions that do not increase the rank of the matroid. We assume that before the insertion $B$ was a base with minimum weight. Note that all subsets of $E$ that were previously independent, are still independent after the addition of $e$ to the matroid. The algorithm outputs the element $f$. In case (a) $w(f) > w(e)$ the maintained base $B'$ is set to $B - f + e$, otherwise in case (b) $w(f) < w(e)$ $B' = B$ remains unchanged. Now assume towards a contradiction that after the insertion there is another base $B''$ such that $w(B'') < w(B')$. Note that $e \in B''$, otherwise $B$ would not have been a minimum weight base. Further, there is an element $e'$ in $C(B+e) - e$ with $B'' - e + e' \in \mathcal{I}$. To see this, let $E_c$ be the set of elements of $c - B''$ for a circuit $c$. If adding an element of $E_{C(B+e)}$ to $B''$ results in a circuit containing $e$, we are done. Otherwise, we know that $e \not \in C(B'' + e'')$ for any $e'' \in E_{C(B+e)}$. Considering such an element $e_1$ from $E_{C(B+e)}$, by \Cref{lm:circuits} we know that there is a circuit $c_1 \subseteq (C(B'' + e_1) + C(B + e)) - e_1$ with $e \in c_1$. If $|E_{c_1}| = 1$ we are done, otherwise (since $c_1 \subseteq B'' + E_{C(B+e)} - e_1$) we can continue by removing the next element $e_2 \in c_1 \setminus B''$ in the same way to find a circuit $c_2 \subseteq (C(B'' + e_2) + c_1) - e_2\subseteq B'' + E_{c_1} - e_2$ with $e \in c_2$. We can continue in the same way, further restricting the number of elements in the next circuit that are not from $B''$ until we get a circuit $c' \subseteq B'' + e_j$ with $e \in c'$. Hence, for the element $e_j \in C(B+e)-e$ we have $B'' - e + e_j$ is a base even before the insertion of $e$. Since $e_j \in C(B+e)$, we have $w(e_j) \leq w(f)$.
    Now, in case (a) where the maintained base was set to $B - f + e$ we get
    \begin{align*}
    w(B''-e+e_j) &= w(B'') - w(e) + w(e_j) < w(B - f + e) - w(e) + w(e_j)\\
    &= w(B) -w(f) + w(e_j) \leq  w(B).
    \end{align*}
    In case (b) where $B$ remained unchanged we have
    \begin{equation*}
    w(B''-e+e_j) = w(B'') - w(e) + w(e_j) < w(B),
    \end{equation*}
    where the last inequality is due to $w(e_j) \leq w(f) < w(e)$. In either case, this contradicts the assumption that $B$ was a base of minimum weight before the insertion.

    Next, we consider the deletion of an element $e$, which restricts the matroid to $E - e$. If the deletion decreases the rank of the matroid, the element cannot be replaced and $B \setminus \{e\}$ gives the new minimum weight base.
    Otherwise, we need to find an element $f$ of minimum weight such that $B \cup \{f\}$ contains a circuit in the matroid before the deletion. Similar to an insertion, the approach is as follows.
    \begin{itemize}
        \item Find min $t$ s.t.\ $\rk(B - e +E_{\leq t}) = \rk(B - e)+1$.
        \item Output the unique element $f$ with weight $t$. 
    \end{itemize} 
    
    Again, assume that before the insertion $B$ was a base with minimum weight and also that $e \in B$ otherwise $B$ stays a minimum weight base. Then $f$ is the element of smallest weight that is not already spanned by $B- e$ (if no such element exists, $e$ was part of all bases before the deletion and $B-e$ is still the minimum weight base). It remains to argue that $B' = B - e + f$ has the smallest total weight after the deletion. Assume there was another base $B''$ for the restricted matroid with $w(B'') < w(B-e+f)$. Then $B'' + e$ contains a cycle $C$ and there is an element $e' \in C$ that is not in $B$. Note that $\{b \in B | w(b) < w(e)\} = \{b \in B'' | w(b) < w(e)\}$ and therefore $w(e') > w(e)$. Now we consider the base for the original matroid $B'' + e - e'$. But then $B$ would not have been a minimum weight base before the deletion, since
    \begin{equation*}
    w(B''+e-e') = w(B'') + w(e) - w(e') < w(B') + w(e) - w(e') \leq w(B') + w(e) - w(f) = w(B)
    \end{equation*}
    where the last inequality follows from $w(f) \leq w(e')$, which remains to be argued. We do so by showing that $e'$ is also an element along with $f$ that is not spanned by $B-e$. Assume towards a contradiction that $B-e$ spans $e'$. Then there is a circuit $C = C(B-e + e')$ where $e'$ has the highest weight out of all elements on the circuit. Since the circuit $c$ also exists in the matroid after the deletion of $e$, $B''$ cannot be a minimum weight base after the deletion, as it contains $e'$. This is because there is an element of smaller weight on $C$ that can be exchanged for $e'$ in $B''$. Similar to the argument in the case of an insertion, we consider the circuit $C(B'' + d)$ for some $d \in C$. If $e'$ is on that circuit we have found the element to exchange $e'$ with. Otherwise there is a circuit $C'$ in $(C + C(B''+d)) - d$ that contains $e'$. If $C'$ contains only one element not from $B''$, we are done. Otherwise we can find a new circuit containing $e'$ with less elements from $C - B''$ as described above until there is only one element left. Hence, $e'$ is not spanned by $B-e$ and $B'$ is the new minimum weight base.

    For the number of queries, note that we can find the minimum with a binary search over the space of all adjusted weights $[1,\dots, n]$. 
\end{proof}
We use \Cref{lm:dyn_min_base} to maintain a greedy base collection. The details are in \Cref{sc:greedyIdeal}.

\newpage
\section{Dynamic Packing and Covering}\label{sc:dynPC} 

In this section, we combine the structural results from \Cref{sc:structural} with the results on (dynamic) base collections from \Cref{sc:greedy,sc:greedyIdeal} to obtain efficient algorithms for base packing and base covering. 
Some of the theorems follow directly by maintaining a greedy base collection of a large enough size. For other results, we use additional techniques specific to the dynamic challenges at hand. 

\subsection{Dynamic Matroid Packing}\label{sc:dynPack}
% \subsubsection{Deterministic Algorithm} 
Next, we combine our results thusfar to obtain our dynamic base packing results. 
By \Cref{remark}, \Cref{lem:greedy_collnew}, and \Cref{lm:dyn_greedy_base_col}, we obtain the worst-case result almost directly. To obtain better bounds with amortization, some more work is required. 

\DynPackDet*
\begin{proof}
    We maintain a base collection $\B$ of $\Theta\left( \frac{\Phi_{\max} \log n}{\epsilon^2}\right)$ greedy bases and maintain the maximum $\ell^{\B}(e)$ using a max-heap directly giving the approximation. 
    
    First, we argue that the maintained value is indeed a $(1\pm \eps)$-approximation of the packing number. We want to show 
    \begin{equation*}
    (1-\eps)\Phi \leq \frac{1}{\max_{e\in E} \ell^\B(e)} \leq (1+\eps)\Phi.
    \end{equation*}
    We show the left hand side, the right hand side is then analogous.
    
    From \Cref{lem:greedy_collnew} we know that if we maintain a base collection of size $\Omega(\Phi \log n/\eps'^2)$, $\ell^\B(e) \geq \ell^*(e) - \frac{\eps'}{\Phi}$ for all $e \in E$ for an $\eps' > 0$ to be defined later, this gives
    \begin{equation*}
    \max_{e \in E} \ell^\B(e) \geq \max_{e \in E} \ell^*(e) - \frac{\eps'}{\Phi}. 
    \end{equation*}
    So for $\eps' := \eps/(1 + \eps)$ we get 
    \begin{equation*}
        \frac{1}{\max_{e\in E} \ell^\B(e)} \leq \frac{1}{\max_{e \in E} \ell^*(e) - \frac{\eps}{(1 + \eps)\Phi}} = \frac{1}{\frac{1}{\Phi} - \frac{\eps}{(1 + \eps)\Phi}} = (1 + \eps) \Phi,
    \end{equation*}
    using that $\max_{e\in E}\ell^*(e)=\Phi^{-1}$, see also \Cref{remark}. Similarly, 
    \begin{align*}
        \frac{1}{\max_{e\in E} \ell^\B(e)} &\ge \frac{1}{\max_{e \in E} \ell^*(e) + \frac{\eps}{(1 + \eps)\Phi}} = \frac{1}{\frac{1}{\Phi} + \frac{\eps}{(1 + \eps)\Phi}} = \left( \frac{1+\eps}{1+\eps+\eps}\right) \Phi\\
        &= \left(1- \frac{\eps}{1+2\eps}\right) \Phi\ge (1-\eps)\Phi,
    \end{align*}

    Next, we consider the number of queries per update. As updating $\B$ can require $|\B|^2$ updates to a greedy base in $\B$ maintaining the base collection takes $O\left( \frac{\Phi_{\max}^2 \log^3(n)}{\epsilon^4}\right)$ queries for any update (worst-case), see \Cref{lm:dyn_greedy_base_col}.

    For an amortized update time, we can reduce this using a technique by \cite{deVosC24}. 
    We note that any element $e$ is in 
    \begin{equation*}
        L^\B(e)=\ell^\B(e)|\B|\le \max_{e'\in E}\ell^\B(e') |\B| \le \frac{|\B|}{(1+\eps)\Phi} \le \frac{|\B|}{\Phi}
    \end{equation*}
    bases. 
    We want to exploit this fact, together with the fact that we only need to consider the first $\Theta(\Phi \log n/\eps^2)$ bases of $\B$ when the base packing number is $\Phi$ by \Cref{lem:greedy_collnew}.
    To exploit this, we partition $\B$ into buckets, where the $i$th bucket contains $2^i$ bases, with $i=0, \dots, \Theta(\log |\B|)$. Intuitively, at any point in time when the packing number is $\Phi$ and $3\Phi\log n/\eps^2 \in (2^i, 2^{i+1}]$, we can process an update $e$ in the bases of the first $i+1$ buckets efficiently -- which is we all we need to output our estimate of $\Phi$. More formally, denote $\Phi$ for the packing number before the update and $\Phi'$ for the packing number after the update.
    
    First, we consider an insertion such that $\Phi' \ge \Phi$. We have $\Phi' \le \Phi+1\le 2\Phi$. Note that after insertion, $e$ will be contained in at most $2\cdot 2^{i+2}/\Phi'= \Theta( \log n/\eps^2)$ bases. When $e$ is part of some base $B_j$, this can lead to at most $|\B|$ recourse in all of $\B$ (see \Cref{lm:dyn_greedy_base_col} for the proof). So inserting $\B$ takes at most $|\B|+ \Theta(|\B|\log n/\eps^2)= \Theta(|\B|\log n/\eps^2)$, where the first factor $\B$ is due to determining which bases to insert $e$ in. 

    It remains to show how to make the changes to bases in the buckets $j>i$. That is where the amortization comes in: rather than performing these updates now. We initialize a priority queue of updates for each bucket. We claim that when a bucket $j$ becomes relevant, i.e., $3\Phi\log n/\eps^2 \in (2^j, 2^{j+1}]$, we can perform all updates from the queue in $\Theta(|\B|\log n/\eps^2)$ queries per update. Hereto, we first perform all insertions from the queue, and then all deletions. 

    Rather than performing one insertion at a time, we consider each base of the matroid, and perform all necessary insertions, where we prioritize the elements with a smaller weight -- this is easily done by maintaining the queue as a min-heap. This means that each inserted element will only be part of $O(\log n/\eps^2)$ bases, since this process is equivalent to computing the greedy bases from scratch. Each such insertion leads to a recourse of $O(|\B|)$. Hence in total, we use $O(|\B|\log n/\eps^2)$ queries per insertion. 

    The insertions now have ensured that the packing number in the matroid is actually high enough, so each element that is deleted, is part of $O(\log n/\eps^2)$ bases. 

    Note that the case that the update is a deletion $e$ such that $\Phi'\le \Phi$ is analogous. 

    We are left with two cases: the update is an insertion and $\Phi' < \Phi$, or the update is a deletion and $\Phi'> \Phi$. Note that for tree-packing in graphs, this does not appear: such updates change the graph from disconnected to connected and the packing number of a disconnected graph is 0. 

    Consider an insertion $e$ such that $\Phi' < \Phi$. The addition of this element needs to increase the overall rank of the matroid. In this case, $e$ is part of \emph{any} base in $\M$ after the insertion. To see this, consider the insertion of an element $e$ and let $E$ be the set of elements before the insertion with packing number $\Phi$. In the following, we do not need to consider sets $S \subseteq E + e$ that do not contain $e$ since in that case we have
        \begin{equation*}
            \frac{|S|}{\rk(E+e) - \rk((E+e) \setminus S)} =\frac{|S|}{\rk(E+e) - \rk((E\setminus S) + e)}  \geq \frac{|S|}{\rk(E) - \rk((E) \setminus S)}.
        \end{equation*}
        This is immediately obvious in the case where $\rk(E+e) = \rk(E)$. Otherwise $\rk(E+e) = \rk(E) + 1$ but in this case $\rk((E\setminus S) + e) = \rk(E\setminus S) + 1$ as well. Hence, these sets cannot be the reason the $\Phi$-value decreases. Now, we want to show that if it decreases, the rank of the matroid must increase. To this end,  assume that $\rk(E+e) = \rk(E)$ and consider any $S \subseteq E$ with $S \neq \emptyset$. Then we have
        \begin{equation*}
            \frac{|S + e|}{\rk(E+e) - \rk((E+e) \setminus (S + e))} =\frac{|S| + 1}{\rk(E) - \rk((E\setminus S))}  > \Phi.
        \end{equation*}
        It remains to consider the set $\{e\}$. By the definition, this cannot be the set giving the new $\Phi'$ as $\rk((E + e) \setminus \{e\}) = \rk(E) = \rk(E + e)$. Hence, the rank of the matroid has to increase when $\Phi$ decreases after the insertion of an element $e$. Further, we have that in this case we also $\Phi' = 1$ and therefore $e$ is in every base. This is because when $\rk(E + e) = \rk(E) + 1$ we have\begin{equation*}
            \frac{|\{e\}|}{\rk(E+e) - \rk(E)} = 1.
        \end{equation*} 
    For any other set $S$ containing $e$ we have that
     \begin{equation*}
            \frac{|S + e|}{\rk(E+e) - \rk((E+e) \setminus (S + e))} =\frac{|S| + 1}{\rk(E) - \rk((E\setminus S)) + 1}  \geq 1.
        \end{equation*}
        where the last inequality follows from $|S| \geq \rk(S)$ and the submodularity of the rank-function.   

        Let $E'$ be the new ground set $E' = E + e$ and let $\M'$ be the matroid after the insertion. Assume that $e$ is not in all bases after the insertion. Then there is a set $B$ that is a base of the matroid before and after the insertion of $e$. Hence, the overall rank does not change i.e. $\rk_\M(E) = \rk_{\M'}(E')$. Now, consider the set $A$ that gives $\Phi' = |A| / (\rk_{M'}(E') - \rk_{\M'}(E' \setminus A)$. The element $e \not \in A$ since otherwise we would have $\rk_{\M'}(E' \setminus A) = \rk_{\M'}(E \setminus (A - e))$ and therefore the set $A - e$ would give a smaller $\Phi'$ compared to $A$. Hence, $\rk_{\M'}(E' \setminus A) \geq \rk_{\M'}(E \setminus A) = \rk_{\M}(E \setminus A)$ and we further get
        \begin{equation*}
            \Phi > \frac{|A|}{\rk_{\M'}(E') - \rk_{\M'}(E' \setminus A)} \geq \frac{|A|}{\rk_{\M}(E) - \rk_{\M}(E \setminus A)}.
        \end{equation*}
        As this would give a smaller packing number than $\Phi$ for the matroid before the insertion, we get that $e$ must be in all bases of the new matroid after the insertion.
    This update is performed easily, since adding to all bases $B\in \B$ can be done in $|\B|$ updates to the corresponding minimum weight base -- where this update has no recourse. So it takes $O(|\B|\log n)$ rank-queries by \Cref{lm:dyn_min_base}. 

    Consider a deletion $e$ such that $\Phi' > \Phi$. As argued for the previous case, this means that $e$ was part of \emph{every} base in $\M$, so we can remove it easily from any $B\in \B$ n $|\B|$ updates to the corresponding minimum weight base -- where this update has no recourse. So it takes $O(|\B|\log n)$ rank-queries by \Cref{lm:dyn_min_base}. 
    
    We obtain a total of $O\left( \frac{\Phi_{\max} \log^3(n) }{\epsilon^4}\right)$ amortized rank-queries for any update. 
\end{proof}

\subsection{Dynamic Matroid Covering}\label{sc:dynCov}
Similar to the previous section, we now present our results for dynamic base covering.

\subsubsection{Deterministic}

First, we provide a deterministic result. 
We note that for base packing, we also have an algorithm with an amortized bound that depends linearly on $\Phi$. For base covering, we could obtain a result in a similar manner that would depend on $\beta_{\max}^2/\Phi$. However, in this case, we cannot relate $\Phi$ and $\beta_{\max}$: $\Phi$ can be constant, even if $\beta_{\max}$ is polynomial. 

\DynCovDet*
\begin{proof}
    Analogously to the packing case, we maintain a base collection $\B$ of $\Theta\left( \frac{\beta_{\max} \log n}{\epsilon^2}\right)$ greedy bases and maintain the minimum $\ell^{\B}(e)$ using a heap directly giving the approximation. 
    
    Again, we argue that the maintained value is a $(1\pm \eps)$-approximation of the covering number. From \Cref{lem:greedy_collnew} we get that
    \begin{equation*}
        (1-\epsilon) \frac{1}{\beta} \leq \min_{e\in E} \ell^B (e)\leq (1+\epsilon) \frac{1}{\beta}.
    \end{equation*}
    By using $\epsilon' := \frac{\epsilon}{1+\epsilon}$ to determine the size of the base collection we get, as for \Cref{thm:dyn_pack_det}: 
    \begin{equation*}
    (1-\eps)\beta \leq \frac{1}{\min_{e\in E} \ell^\B(e)} \leq (1+\eps)\beta.
    \end{equation*}  
    
    Maintaining the base collection takes $O\left( \frac{\beta_{\max}^2 \log^3(n)}{\epsilon^4}\right)$ queries for any update (worst-case), see \Cref{lm:dyn_greedy_base_col}. 
\end{proof}

\subsubsection{Oblivious adversary}
Next, we show how uniform sampling can help us to obtain update times independent of $\beta$ against oblivious adversaries. The approach uses a standard sampling technique (see, e.g., \cite{McGregorTVV15}) and is similar to the sampling to maintain arboricity in multigraphs, as shown in~\cite{deVosC24}.
\DynCov*
\begin{proof}
We create $O(\log n)$ matroids $\M_i$ by sampling each element with probability ${p_i = \frac{24c \log n}{2^i \eps^2}}$ (for some constant $c \ge 2$ to be set later and $i$ such that $p_i < 1$) to get the respective universe $E_i$. $\M_i$ is then given by the restriction of the original matroid to the sampled elements $\M_i=\M|E_i$. Now we want to show that if $\beta \in [2^{i-1}, 2^{i+2})$ then $\frac{1}{p_i}\beta_i$ is a $(1\pm \eps)$-approximation of $\beta$. We show this in two parts.

First, we argue that $\frac{1}{p_i}\beta_i \geq (1-\eps) \beta$. Consider a set $S \subseteq E$ with $\beta = \frac{|S|}{\rk_{\M}(S)}$. Let $S_i = S \cap E_i$, then we get
\begin{equation*}
\Pr\left( \frac{|S_i|}{p_i \rk_{\M_i}(S_i)} < (1-\eps) \beta\right) = \Pr\left( |S_i| < (1-\eps) p_i \beta \rk_{\M_i}(S_i) \right).
\end{equation*}
Since $\rk_{\M_i}(S_i) \leq \rk_\M(S) = \frac{|S|}{\beta}$ and $\mathbb{E}(|S_i|) = p_i |S|$ we get by a Chernoff bound
\begin{align*}
\Pr\left( |S_i| < (1-\eps) p_i \beta \rk_{\M_i}(S_i) \right) 
&\le \Pr\left( |S_i| < (1-\eps) p_i \beta \rk_{\M}(S) \right)
=\Pr\left( |S_i| < (1-\eps) p_i |S| \right)\\
&\leq e^{-\eps^2 p_i |S|/2} =  e^{-\eps^2 p_i \beta \rk_\M(S)/2} \leq e^{-\eps^2 p_i 2^{i-1} \rk_\M(S)/2}\\
&\leq n^{-c \rk_\M(S)}.
\end{align*}

Hence 
\begin{align*}
\Pr\left( \frac{|S_i|}{p_i \rk_{\M_i}(S_i)} < (1-\eps) \beta\right) \leq n^{-c \rk_\M(S)}.
\end{align*}

Now note that $\beta_i= \max_{\emptyset \neq S'\subseteq E_i} \frac{|S'|}{\rk_{\M_i}(S')} \geq \frac{|S_i|}{\rk_{\M_i}(S_i)}$, so w.h.p.\ $\frac{1}{p_i}\beta_i \geq (1-\eps) \beta$.

Second, we need to show that with high probability the covering number is bounded as follows $\frac{1}{p_i}\beta_i \leq (1+\eps) \beta$. 

Let $S_i \subseteq E_i$ be any non-empty subset. We will show that w.h.p.\ $\frac{|S_i|}{p_i\rk_{\M_i}(S_i)}\leq (1+\eps)\beta$.
%--
Hereto, we consider each possible rank $1 \leq r \leq \rk_{\M}(E_i) = \rk(\M |E_i)$ of a non-empty subset of $E_i$ separately. Let $S_i$ be the subset of $E_i$ with rank $r$ with maximum $|S_i|$. Note that for any other set $A \subseteq E_i$ with the same rank we have $\frac{|A|}{p_i\rk_{\M_i}(A)} \leq \frac{|S_i|}{p_i\rk_{\M_i}(S_i)}$.

Now we consider the set $S := \Span_{\M}(S_i)$. For this set we have (a) $\rk_{\M}(S) = \rk_{\M}(S_i) = \rk_{\M_i}(S_i)$ since $S_i \subseteq E_i$ and (b) $S \cap E_i = S_i$ since no element of $E_i$ can be added to $S_i$ without an increase in the rank. Therefore we get
\begin{align*}
\Pr\left( \frac{|S_i|}{p_i \rk_{\M_i}(S_i)} > (1+\eps) \beta\right) &= \Pr\left( \frac{|S \cap E_i|}{p_i \rk_{\M}(S)} > (1+\eps) \beta\right)\\ &= \Pr\left( |S \cap E_i| > (1+\eps) p_i \beta \rk_{\M}(S)\right).
\intertext{Further, since $\beta \rk_\M(S) \geq |S|$, using a Chernoff bound we get}
\Pr\left( |S \cap E_i| > (1+\eps) p_i \beta \rk_{\M}(S)\right) &\leq e^{-\eps^2 p_i \beta \rk_\M(S)/3} =  e^{-\eps^2 24 c \log (n) \beta \rk_\M(S)/3 \eps^2 2^i}\\
&\leq e^{-4c\log (n) \rk_\M(S)} \leq n^{-c (\rk_\M(S) + 2)}
\end{align*}
where the second to last inequality follows since $\beta \geq 2^{i-1}$. Note that for such a set $S$ we have that $S = \Span_{M}(S)$. Also, among the sets of maximum size with rank $r$ the set $S$ is unique. If there was another set $S_j \neq S_i$ with $\Span(S_j) = S$ , then $S_j \subseteq S$ and we would have an element $e \in S_j \setminus S_i$ and $|S_i + e| > |S_i|$ while $\rk_\M(S_i + e) = \rk_M(S_i)$. This contradicts the choice of $S_i$. Hence, for any rank $r$ it suffices to union bound over all sets $S$ with $S = \Span(S)$ and $\rk(S)=r$, there are at most $n^r$ such sets. So the statements holds for all such $S$ with probability $n^{-2c}$.

This allows us to then union bound over all possible ranks $r$ , there are at most $\rk_{\M}(E)$ many and $n$ updates, so the statement holds with probability $n^{-2c}\cdot n \cdot n = n^{2-2c}\leq n^{-c}$.

Hence, we get that $\beta_i \leq (1+ \eps) \beta p_i \leq 2^3 \frac{24 c \log n}{\eps^2} = O\left(\frac{\log n}{\eps^2}\right)$ when $\beta \leq 2^{i+2}$. The algorithm maintains a greedy base collection of size $\Theta\left(\frac{\log^2 n}{\eps^4}\right)$ for each of the $O(\log n)$ matroids $\M_i$ and one for $\M$ itself. Whenever $\beta \in [2^{i-1}, 2^{i+2})$ the collection corresponding to $\M_i$ gives the correct approximation. All other collections might give values that differ drastically from $\beta$ and are simply disregarded. At any point in time to know which $\M_i$ gives the correct value it suffices to consider the estimate from the previous update. This is because if the old estimate $\frac{\beta_i}{p_i}$ was in $[2^i, 2^{i+1})$, then after the update we have that $\beta < (1+\epsilon) \frac{\beta_i}{p_i} + 1 \leq 2^{i+2}$ and similarly $\beta \geq (1-\eps) \frac{\beta_i}{p_i} - 1 \geq 2^{i-1}$. If the current estimate is in $\Theta(\log n / \eps^2)$ the estimate on the original matroid is considered, otherwise the interval $[2^{i}, 2^{i+1})$ that contains the previous estimate determines the base collection $\M_i$ that is used for the estimate after the next update. Now it remains to analyze the runtime. For each matroid $\M_i$ we maintain a greedy base collection of size $O\left(\frac{\log^2 n}{\eps^4}\right)$. 

There are $O(\log n)$ matroids that could be affected by an update, so by \Cref{lm:dyn_greedy_base_col}, an update requires $O\left(\frac{\log^6(n)}{\eps^8}\right)$ rank-queries in the worst-case.
\end{proof}

\newpage

\section{Packing, Covering, and Base Collections}\label{sc:structural}
In this section, we provide some structural results regarding base packing, base covering, and their relationship to base collections. These results are the foundation for the algorithms of \Cref{sc:dynPC}. 

\subsection{Base Packing}
First, we examine how the largest possible relative load of any element in all possible base collections relates to the packing number. Recall that the following theorem gives the integral packing number of a matroid.
\begin{theorem}[\cite{Edmonds65}]\label{thm:matroids_edmonds}
    A matroid $\mathcal{M}$ can pack $k$ disjoint bases if and only if for every subset $A\subseteq E$ we have $
        |A|\ge k(\rk(E)-\rk(\overline A))$.
\end{theorem}

To relate this to the graphic matroid, let $\P$ be a vertex partition of a graph $G=(V,E)$. Then the value $ \frac{|A| }{\rk(E)-\rk(\overline A)}$ corresponds to $ \frac{|E(G/\P)| }{|\P|-1}$, called the \emph{partition value} of $\P$.

Next, we use this theorem to show a duality between this value and base collection. 

\matroidNW*
\begin{proof}
    We prove the two inequalities "$\le$" and "$\ge$", starting with the former, which follows directly without using \Cref{thm:matroids_edmonds}. 

    "$\le$". Let $A^*\subseteq E$ be such that $\frac{|A^*| }{\rk(E)-\rk(\overline {A^*})}=\min\limits_{\substack{A\subseteq E\ \rm{s.t.}\\ \rk(\overline A) < \rk(E)}}  \frac{|A| }{\rk(E)-\rk(\overline A)}$. Now we note that any base $B$ needs to contain $\rk(E)-\rk(\overline {A^*})$ elements from $A^*$. So for any base collection $\B$, we have $\sum_{e\in A^*}\ell^\B(e)\geq \rk(E)-\rk(\overline {A^*})$. This means that on average for $e\in A^*$ we have $\ell^\B(e)\geq \frac{\rk(E)-\rk(\overline {A^*})}{|A^*|}$. In particular $\max_{e\in E}\ell^\B(e) \geq \max_{e\in A^*} \ell^\B(e)\geq \frac{\rk(E)-\rk(\overline {A^*})}{|A^*|}$. Equivalently, we get that for any base collection $\B$
    \begin{equation*}
        \frac{1}{\max_{e\in E}\ell^\B(e)} \le \frac{|A^*| }{\rk(E)-\rk(\overline {A^*})}.
    \end{equation*}
    Now we see that 
    \begin{equation*}
        \max_{\mathcal B}\frac{1}{\max_{e\in E} \ell^\B(e)} \le \frac{|A^*| }{\rk(E)-\rk(\overline {A^*})}= \min_{\substack{A\subseteq E\ \rm{s.t.}\\ \rk(\overline A) < \rk(E)}}  \frac{|A| }{\rk(E)-\rk(\overline A)}.
    \end{equation*}

    "$\ge$". W.l.o.g., we assume that $\min\limits_{\substack{A\subseteq E\ \rm{s.t.}\\ \rk(\overline A) < \rk(E)}}  \frac{|A| }{\rk(E)-\rk(\overline A)}$ is an integer. This can be achieved by copying each element in the matroid $\rk(E)!$ times, since $\rk(E)-\rk(\overline {A^*})\in \{1,2,\dots \rk(E)\}$.
    This has no impact on the analysis as it blows up the values on the left and right equally.
    
    Now suppose that $\min\limits_{\substack{A\subseteq E\ \rm{s.t.}\\ \rk(\overline A) < \rk(E)}}  \frac{|A| }{\rk(E)-\rk(\overline A)}=k$, by \Cref{thm:matroids_edmonds} this means that we can pack $k$ disjoint bases. Let $\B$ be the base collection consisting of these $k$ bases. Then for each element $e\in B\in \B$, we have $\ell^\B(e)=1/k$. For other elements, we have $\ell^\B(e)=0$. Hence we have $\max_{e\in E}\ell^\B(e)=1/k$. Equivalently
    \begin{equation*}
        \frac{1}{\max_{e\in E} \ell^\B(e)} = k = \min_{\substack{A\subseteq E\ \rm{s.t.}\\ \rk(\overline A) < \rk(E)}}  \frac{|A| }{\rk(E)-\rk(\overline A)}.
    \end{equation*}
    Hence clearly 
    \begin{equation*}
        \max_{\B}\frac{1}{\max_{e\in E} \ell^\B(e)} \ge k = \min_{\substack{A\subseteq E\ \rm{s.t.}\\ \rk(\overline A) < \rk(E)}}  \frac{|A| }{\rk(E)-\rk(\overline A)}.\qedhere
    \end{equation*}
\end{proof}

However, for purpose of analysis, we would like a \emph{specific} packing with this property. More precisely, we would like the \emph{loads} corresponding to that packing. Hence, we use the ideal relative loads. 

\paragraph{Ideal Relative Loads.}
Next, we want to show that the ideal relative loads are well-defined. Hereto, we first show that they will be non-decreasing. 

\begin{lemma}
\label{lem_phi_nondec}
    Let $\mathcal{M}_0 = \mathcal{M}$, and for $i\ge 1$ let $\mathcal{M}_i$ be the matroid considered in the $i$-th recursive step while assigning the ideal relative loads. The values of $\Phi$ are non-decreasing, meaning $\Phi_{\mathcal{M}_i} \leq \Phi_{\mathcal{M}_{i+1}}$ for all $i \geq 0$.
\end{lemma}

\begin{proof}
    Consider an arbitrary step $i$ where the matroid $\mathcal{M}_i = (E,\mathcal{I})$ is processed and let in the following $\Phi := \Phi_{\M_i}$. Let further $A_0$ be a set that gives $\Phi = |A_0|/(\rk(E) - \rk(E \setminus A_0))$. Assume that there is a set $A_1 \subset E \setminus A_0$ such that $\frac{|A_1|}{\rk(E \setminus A_0) - \rk((E \setminus A_0) \setminus A_1)} < \Phi$, or equivalently $\rk(E \setminus A_0) - \rk((E \setminus A_0) \setminus A_1) > |A_1|/\Phi$. By definition of $\Phi$, we also have $\frac{|A_0|}{\rk(E) - \rk(E \setminus A_0)} = \Phi$, or equivalently: $\rk(E) - \rk(E \setminus A_0)= |A_0|/\Phi$. 
    
    Since $A_0 \cap A_1 = \emptyset$, we get
        \begin{align*}
            \frac{|A_1 \cup A_0|}{\rk(E) - \rk(E \setminus (A_0 \cup A_1))} &= \frac{|A_0| + |A_1|}{\rk(E) - \rk(E \setminus A_0) + \rk(E \setminus A_0) - \rk((E \setminus A_0) \setminus A_1)}\\
            &< \frac{|A_0|+|A_1|}{\frac{|A_0|}{\Phi}+\frac{|A_1|}{\Phi}} = \Phi, 
        \end{align*}
    which contradicts the choice of $|A_0|$.
\end{proof}

Now we are ready to show the following lemma. 
\begin{lemma}\label{lm:ideal_welldefined}
    The ideal relative loads are well-defined. 
\end{lemma}
\begin{proof}
    Suppose we have two sets $S$ and $T$ such that $T \not \subseteq S$ and
    \begin{equation}\label{eq:A0A1}
        \frac{|S|}{\rk(E)-\rk(E\setminus S)} = \Phi = \frac{|T|}{\rk(E)-\rk(E\setminus T)}.
    \end{equation}
    We need to show that all elements  $e\in T$ get value $\ell^*(e)=1/\Phi$, even if we first recurse on $E\setminus S$. In other words, we need to show that 
    \begin{equation*}
     \frac{|T\cap (E\setminus S)|}{\rk(E\setminus S)-\rk((E\setminus S)\setminus (T\cap (E\setminus S)))}= \Phi.
    \end{equation*}
    First, we note that $\rk(E\setminus S)-\rk((E\setminus S)\setminus (T\cap (E\setminus S))) =  \rk(E\setminus S) - \rk((E\setminus S) \cap (E\setminus T))\neq 0$ as otherwise we would get that $ \rk((E\setminus S) \cup (E\setminus T))= \rk(E \setminus (S \cap T)) \leq \rk(E \setminus T)$. Then it would follow that $\frac{|T \cap S|}{\rk(E)-\rk(E\setminus (T \cap S))} < \Phi$ which contradicts the definition of $\Phi$.
    By \Cref{lem_phi_nondec}, we already know that it is at least $\Phi$. So it remains to show that it is at most $\Phi$. For contradiction, assume 
    \begin{equation}\label{eq:contr_assump}
     \frac{|T\cap (E\setminus S)|}{\rk(E\setminus S)-\rk((E\setminus S)\setminus (T\cap (E\setminus S)))}> \Phi.
    \end{equation}

    Then we will show that 
    \begin{equation*}
        \frac{|T\setminus (E\setminus S)|}{\rk(E)-\rk(E \setminus (T\setminus (E\setminus S)))} < \Phi,
    \end{equation*}
    contradicting the value of $\Phi$, which was minimal. 
    First, we consider the numerator: 
    \begin{align*}
        |T\setminus (E\setminus S)| &= |T|-|T\cap (E\setminus S)|\\
        &< 
        \Phi (\rk(E)-\rk(E\setminus T)) - \Phi( \rk(E\setminus S)-\rk((E\setminus S)\setminus (T\cap (E\setminus S)))),
        \intertext{by \Cref{eq:A0A1} and \Cref{eq:contr_assump}}
        &= \Phi (\rk(E)-\rk(E\setminus T)- \rk(E\setminus S)+\rk((E\setminus S)\cap (E\setminus T)))\\
        &\le \Phi( \rk(E) -\rk( (E\setminus S)\cup (E\setminus T))),
    \end{align*}
    by submodularity of the rank function.
    Finally, we note that 
    \begin{equation*}
         (E\setminus S)\cup (E\setminus T) = E\setminus (S\cap T) =E \setminus (T\setminus (E\setminus S)). 
    \end{equation*}
    So we conclude that 
  \begin{equation*}
        \frac{|T\setminus (E\setminus S)|}{\rk(E)-\rk(E \setminus (T\setminus (E\setminus S)))} < \frac{\Phi(\rk(E)-\rk(E \setminus (T\setminus (E\setminus S))))}{\rk(E)-\rk(E \setminus (T\setminus (E\setminus S)))} = \Phi,
    \end{equation*}
    which finishes the proof. 
\end{proof}

\subsection{Base Covering} 
As shown in \cite{Edmonds65partition}, the number of bases that are needed to cover all elements of $e$ are given by 
\begin{equation*}
    \left \lceil \max_{\substack{A\subseteq E\ \rm{s.t.} \\ A \neq \emptyset}} \frac{|A|}{\rk(A)} \right \rceil.
\end{equation*}

To obtain a similar result as for base packing, we relate the optimal base collection to the covering number. This generalizes the arboricity result of de Vos and Christiansen~\cite[Theorem 25]{deVosC24}, which is the statement in the special case of graphic matroids. 
\matroidPacking*
\begin{proof}
    Let $E_i := E \setminus \bigcup_{j < i} A_j$ and let $\M_i = \M|E_i$ for all recursion levels denoted by the indices $i \in I$. First, we note that $ 1/ \min_e\ell^*(e) = \max_{i \in I} \Phi_{\M_i}$ over all recursion levels $i$. Let $j$ be the last iteration in which $\Phi_{\M_i}$ is maximized. Since the $\Phi$-values are non-decreasing, see \Cref{lem_phi_nondec}, this also has to be the last level of the recursion and therefore $E_{j+1} = E_j \setminus A_j = \emptyset$ and $\rk(E_j) = \rk(A_j)$. We get
    \begin{equation*}
        1/ \min_e\ell^*(e) = \frac{|A_j|}{\rk(E_j) - \rk(E_j \setminus A_j)} = \frac{|A_j|}{\rk(A_j)} \leq \max_{\substack{A\subseteq E\ \rm{s.t.} \\ A \neq \emptyset}} \frac{|Y|}{\rk(Y)}.
    \end{equation*}
    Now, we go on to show $1/ \min_e\ell^*(e) \geq \frac{|Y|}{\rk(Y)}$ for any subset $Y \subseteq E$. Note that the recursively defined sets $A_i$ for $i \in I$ are disjoint and that $\bigcup_{i \in I} A_i = E$. Hence, we get
    \begin{equation*}
        |Y| = \sum_{i \in I} |A_i \cap Y|.
    \end{equation*}
    Next, we want to bound $|Y|$ by $\sum_{i\in I} \Phi_{\M_i} \rk(A_i \cap Y)$
    where it remains to show the inequality $|A_i \cap Y| \leq \Phi_{\M_i} \rk(A_i \cap Y)$ for any $i \in I$. Assume towards a contradiction that there is a level $k \in I$ for which the inequality does not hold. In the following let $\overline{A_i} := E_i \setminus A_i$ for each $i \in I$. Now, consider the value
    \begin{align*}
        \Phi' &= \frac{|A_k \setminus Y|}{\rk(E_k) - \rk(\overline{A_k}\cup (A_k \cap Y))} = \frac{|A_k| - |A_k \cap Y|}{\rk(E_k) - \rk(\overline{A_k}\cup (A_k \cap Y))}.
        \intertext{Since the rank function is submodular we get}
        \Phi' &\leq \frac{|A_k| - |A_k \cap Y|}{\rk(E_k) - \rk(\overline{A_k}) - \rk(A_k \cap Y)}.
        \intertext{By our assumption that $|A_k \cap Y| > \Phi_{\M_k} \rk(A_k \cap Y)$ and the definition of $\Phi_{\M_k}$ we further get}
        \Phi' &< \frac{|A_k| - \Phi_{\M_k} \rk(A_k \cap Y)}{\rk(E_k) - \rk(\overline{A_k}) - \rk(A_k \cap Y)} = \frac{\Phi_{\M_k}(\rk(E_k) - \rk(\overline{A_k})) - \Phi_{\M_k} \rk(A_k \cap Y)}{\rk(E_k) - \rk(\overline{A_k}) - \rk(A_k \cap Y)}.
    \end{align*}
    This contradicts the choice of $A_k$, as picking $A_k \setminus Y$ in recursion level $k$ would result in a smaller $\Phi$-value. Next, we show that $\sum_{i \in I} \rk(A_i \cap Y)\leq \rk(Y)$. Note that for every $i \in I$ with $i \neq 0$ we have $E_i = \overline{A_{i-1}}$. Hence, for the last level of recursion $i_{\text{max}}$, we get
    \begin{equation*}
        \sum_{i \in I} \rk(E_i \cap Y) - \rk(\overline{A_i} \cap Y) = \left(\sum_{i \in I, i < i_{\text{max}}} \rk(E_i \cap Y) - \rk(E_{i+1} \cap Y)\right) + \rk(E_{i_{\text{max}}} \cap Y) - \rk(\emptyset) \leq \rk(Y)
    \end{equation*}
    Now it remains to show that $\rk(A_i \cap Y)\leq \rk(E_i \cap Y) - \rk(\overline{A_i} \cap Y)$ for all $i \in I$. In the following, consider an arbitrary $i \in I$. Note that if $\Phi_{\M_i}$ is an integer then
    \begin{equation}\label{eq_ub_rk_aiy}\rk(A_i \cap Y) \leq \frac{|A_i \cap Y|}{\Phi_{\M_i}}
    \end{equation} since the matroid $\M_i$ contains $\Phi_{\M_i}$ disjoint bases. Hence, every base of $\M_i | (A_i \cap Y)$ contains at most $|A_i \cap Y|/ \Phi_{\M_i}$ many elements. Next, we want to show 
    \begin{equation}\label{eq_ub_by_phi}\frac{|A_i \cap Y|}{\rk(E_i \cap Y) - \rk(\overline{A_i} \cap Y)} \leq \Phi_{\M_i}.
    \end{equation}
    Assume the inequality did not hold. Similar to the approach above, consider the value
    \begin{align*}
        \Phi'_{\M_i} &= \frac{|A_i \setminus Y|}{\rk(E_i) - \rk(\overline{A_i}\cup (A_i \cap Y))} = \frac{\Phi_{\M_i}(\rk(E_i) - \rk(\overline{A_i})) - |A_i \cap Y|}{\rk(E_i) - \rk(\overline{A_i}\cup (A_i \cap Y))}.
    \intertext{By the assumption we get}
        \Phi'_{\M_i} &< \frac{\Phi_{\M_i}(\rk(E_i) - \rk(\overline{A_i}) - \rk(E_i \cap Y) - \rk(\overline{A_i} \cap Y))}{\rk(E_i) - \rk(\overline{A_i}\cup (A_i \cap Y))}.
    \end{align*}
    Now, in order to get $\Phi'_{\M_i} < \Phi_{\M_i}$ we need
    \begin{align*}
        \rk(\overline{A_i}\cup (A_i \cap Y)) &\leq \rk(\overline{A_i}) +  \rk(E_i \cap Y) - \rk(\overline{A_i} \cap Y).
    \intertext{This inequality holds since we have} 
        \rk(\overline{A_i}\cup (A_i \cap Y)) + \rk(\overline{A_i} \cap Y) &= \rk(\overline{A_i}\cup (E_i \cap Y)) + \rk(\overline{A_i} \cap (E_i \cap Y))\\
        &\leq \rk(\overline{A_i}) + \rk(E_i \cap Y)
    \end{align*}
    where the last inequality is due to the submodularity of the rank function.
    Now we have that \eqref{eq_ub_rk_aiy} and \eqref{eq_ub_by_phi} hold for any $i \in I$. We can again assume w.l.o.g.\ that $\Phi_{\M_i}$ is an integer using the same technique as described in the proof of \Cref{cor_phi}, as the rank remains unchanged. Putting \eqref{eq_ub_rk_aiy} and \eqref{eq_ub_by_phi} together, we get the desired property which directly gives the following lower bound on the rank of $Y$
    \begin{equation*}
        \sum_{i \in I} \rk(A_i \cap Y) \leq \rk(Y).
    \end{equation*}
    Finally, we get: 
     \begin{align*}
        \frac{|Y|}{\rk(Y)} &\leq \frac{\sum_{i \in I} \Phi_{\M_i} \rk(A_i \cap Y)}{\sum_{i \in I} \rk(A_i \cap Y)}\\
        &\leq \max_{i \in I} \Phi_{\M_i} \\
        &= \frac{1}{\min_e\ell^*(e)}.\qedhere
    \end{align*}    
    
\end{proof}

\newpage
\section{Greedy Base Collection}\label{sc:greedyIdeal}

First, let us define a greedy base collection as follows:
let the weight of an element be the number of bases
an element belongs to, and require that the base in the collection form successive minimum weight bases. This is a natural generalization of the notion for graphic matroids called `greedy tree-packing'\footnote{We opt for `collection' instead of `packing' to avoid any confusion with base packing as in \Cref{thm:matroids_edmonds}.}, see, e.g.,~\cite{Thorup07}.

In this section, we show how to maintain a greedy base collection. Then, we show that when we compute a greedy base collection $\B$ that is large enough, it approximates the ideal relative loads well.

\subsection{Dynamic Greedy Base Collection}
We show how to maintain a greedy base collection under dynamic updates. This lemma is rather straightforward.

\GreedyBaseCol*
\begin{proof}
    Given a greedy base collection $\B$, let there be an deletion or insertion~$e$. First, consider the deletion of $e$. The element $e$ is part of at most $|\B|$ bases. To update our data structures, we need to delete it from every base it is part of. We do this in sequence and show that each such deletion leads to at most $|\B|$ updates to a base. 
    Since each update to a minimum weight base uses at most $O(\log(n))$ queries per update by \Cref{lm:dyn_min_base}, 
    we get $O(|\B|^2\log(n))$ worst-case queries per update in total. 
    
    In the following, for any element $f \in E$ and $i \in \{1, \ldots, |\mathcal{B}|\}$, let $w_i(f)$ be the number of bases $B_j$ that contain $f$ with $ j \le i$. 

    It remains to show how to delete $e$ from a single base $B_{i}$ in $\B$. We do this by first updating $B_{i}$, then $B_{i+1}$ and so forth. We claim (by induction) that during this process, at any base $B_j$ (for $j\ge i_e$) we only need to make one update, i.e., increase the weight of one element $e'$ and potentially replace it.  
    
    Let $B_j\in \B$ be the base that where we need to increase the weight of $e'$. We claim that we get at most $|\B|-j$ changes to other bases. Note that this is non-trivial; at a first glance, the changes over bases could cascade. However, we note that if deleting $e'$ from $B_j$ leads to adding $f$ to $B_j$, then $f$ gets an increased weight $w_j^{\text{new}}(f)=w_j^{\text{old}}(f)+1$. Now consider the next base $B_k$ with $k>j$ that contains $f$. We have two cases: either $f$ remains in the base $B_k$, even though its weight is increased. This means we have $w_k^{\text{new}}(f)=w_k^{\text{old}}(f)+1$ and all other elements still have the same weight. Alternatively, an element $g$ is replacing $f$ in $B_k$. This means that $w_k^{\text{new}}(g)=w_k^{\text{old}}(g)+1$ and $w_k^{\text{new}}(f)=w_k^{\text{old}}(f)+1-1=w_k^{\text{old}}(f)$. So for bases $B_l$ for $l>k$, we only need to adjust the weight of $g$. 
    Hence, there can be at most one change per base in the collection, so $|\B|$ in total. The proof for an insertion of $e$ is analogous.
\end{proof}

\subsection{Greedy Approximates Ideal Base Collection}
Now, we show that when we compute a greedy base collection $\B$ that is large enough, its relative loads $\ell^\B(\cdot)$ give good approximations to the ideal relative loads $\ell^*(\cdot)$. This generalizes the result for packing trees \cite{Thorup07} to matroids, so the proof follows the approach used for the graph case. It utilizes a technique by Young \cite{young95}. The idea is to consider a distribution of bases such that picking the bases for a base collection randomly from this distribution would result in the relative loads being ideal in expectation. Now, the number of violations, where the relative load of an element does not approximate its ideal load well, is analyzed. This is done while replacing the randomly picked bases one after the other by bases computed greedily. For the analysis, we consider a pessimistic estimator for the number of violations at every step of the process and show that its value cannot increase when a greedy base is added. 

Now, we define $\Pi_\mathcal{M}$, a probability distribution of bases of $\mathcal{M}$. We construct $\Pi_\mathcal{M}$  simultaneously with the assignment of the ideal relative loads. See the process that defines ideal relative loads in \Cref{sec:overview_base_collactions}. On each (restricted) matroid $\M$ in this process, let $\B^*$ be a base collection for $\M$ with $\frac{1}{\max_{e\in E} \ell^{\B^*}(e)} = \Phi_\M$. To pick a random base from $\Pi_\M$, pick a base $B$ from $\B^*$ uniformly at random. 

In the following, we generalize the proof of Lemma 13 from~\cite{Thorup07}. 
\begin{lemma}\label{lem_prob}
    For each $e \in E$, we have \[\underset{R \in \Pi_\mathcal{M}}{\Pr}(e \in R) = \ell^*(e). \]
\end{lemma}

\begin{proof}
First, pick a random base from $\Pi_\M$. Then, pick an independent set $B'$ randomly from the recursively defined distribution $\Pi_{\M|\overline{A_0}}$. We construct a base $B'' := (B \cap A_0) \cup B'$. First, we show that $B''$ is indeed a base for $\M$.

By the exchange property there is a set $S \subseteq B$ such that $B' \cup S \in \mathcal{I}$ and $|B' \cup S| = |B|$, making this set a base. It remains to show that $S = B \cap A_0$. Since $B'$ was a base for $\M|\overline{A_0}$ and therefore spans $\overline{A_0}$, we have $S \subseteq B \cap A_0$ with $|S| = \rk(E) - \rk(\overline{A_0})$. Any set in $\B^*$ is a base for $\M$. Hence, every such base contains at least $\rk(E) - \rk(\overline{A_0})$ elements from $A_0$, giving a lower bound on the total load of all elements in $A_0$. So, the average relative load over $A_0$ is then at least $\frac{\rk(E) - \rk(\overline{A_0})}{|A_0|} = \frac{1}{\Phi_\mathcal{M}}$. This equals the maximum relative load of elements in $A_0$. Hence, all relative loads in $A_0$ are the same and no base in $\B^*$, including $B$, can contain more than $\rk(E) - \rk(\overline{A_0})$ elements from $A_0$. Hence, we get $S = B \cap A_0$ and $B''$ is indeed a base for $\M$.

Further, for an element $e \in A_0$ we have 
\begin{equation*}
    \Pr_{B'' \in \Pi_\M}(e\in B'') = \Pr_{B\in B^* \rm{ u.a.r.}}(e\in B) = \ell^{\B^*}(e) = 1/\Phi_\M= \ell^*(e).
\end{equation*}
For an element $e \in \overline{A_0}$ we instead get 
\begin{equation*}\Pr_{B'' \in \Pi_\M}(e\in B'') = \Pr_{B' \in \Pi_{\M|\overline{A_0}}}(e\in B') = \ell^*(e)
\end{equation*}
where the last equality follows by induction.
\end{proof}

 Note that any $B \in \Pi_\M$ is a minimum weight base with respect to the ideal relative loads $\ell^*(\cdot)$ as by induction $B \cap \overline{A_0}$ is given by a minimum weight base picked from $\Pi_{\M|\overline{A_0}}$. Further, the $\ell^*(\cdot)$-values of elements in $\overline{A_0}$ are smaller or equal compared to elements in $A_0$.

 Now, we go on to show that a large enough greedy base collection approximates the ideal one.

\GreedyColNew*

We first show that \Cref{eq:greedy_coll_low} holds for all $e \in E$ with $\ell^*(e) \leq 1/\gamma$ and then argue that we have \Cref{eq:greedy_coll_high} for all $e \in E$ with $\ell^*(e) \geq 1/\gamma$. The proof follows the structure of the special case~$\eqref{eq:greedy_coll_low}$ of the graphic matroid, i.e., trees, for $\gamma = \Phi$ from \cite[Proposition 16]{Thorup07} using the distribution of bases from above. Many parts are direct generalizations for matroids; all are included here for completeness.

We show the statement in two separate parts, proving $\ell^{\B}(e) \leq \ell^*(e) + \eps/\gamma$ and $\ell^{\B}(e) \ge \ell^*(e) - \eps/\gamma$ using the following approach. We consider an estimator for the number of violations for the statement above, while we construct the greedy base collection one base at a time. We assume that the remaining bases are from the distribution $\Pi_\M$. The proof is comprised of showing the following three properties: (i) initially the value of the estimator is below $1$, (ii) adding a greedy base does not increase the value, and (iii) having a value below $1$ at the end of the process implies the desired property for each element. The proofs of the first two properties closely follow the corresponding parts from \cite{Thorup07}, we include the version for matroids for completeness. The main difference is in the proof of the last property.
\begin{claim}
    For all elements $e \in E$ with $\ell^*(e)\le 1/\gamma$ we have
    \begin{equation}\label{eq_upperbound}
        \ell^{\B}(e) \leq \ell^*(e) + \eps/\gamma.
    \end{equation}
\end{claim}
\begin{proof}
    We show that $\ell^{\B}(e) \leq d + \eps/\gamma$ holds for all elements $e \in E$ with $\ell^*(e) \leq d $ for an arbitrary $d \leq 1/\gamma$. Consider the set $A$ containing all such elements for the given $d$, $A := \{ e \in E \;|\;\ell^*(e) \leq d\}$. Now, we analyze the change of the value of the estimator, where $t$ denotes the total number of bases in the collection at the end of the process,
    \begin{equation}\label{eq_estimator} \sum_{e \in A}\frac{(1+\epsilon)^{L^{\B}(e)}(1+\epsilon d)^{t-|\B|}}{(1+\epsilon)^{d t+\epsilon t/\gamma}}.
\end{equation}
First, we show (i) that in the beginning, when the greedy base collection is still empty, the value of the estimator is below $1$. When $\B = \emptyset$, the estimator \eqref{eq_estimator} becomes
 \begin{equation*} \sum_{e \in A}\frac{(1+\epsilon d)^{t}}{(1+\epsilon)^{d t+\epsilon t/\gamma}}.
\end{equation*}
Analogous to \cite{Thorup07} we get
\begin{equation*}
\sum_{e' \in A}\frac{(1+\epsilon d)^{t}}{(1+\epsilon)^{d t+\epsilon t/\gamma}} < n\left(\frac{e^{\epsilon}}{(1+\epsilon)^{(1+\epsilon)}}\right)^{t/\gamma} \leq ne^{-\epsilon^2t/3\gamma} = ne^{- \log n} = 1.
\end{equation*}
The inequalities follow from  $(1+\eps d) \leq e^{\eps d}$, $d \leq 1/\gamma$ and $t = 3 \gamma \log n / \eps^2$.

Now we consider the end of the greedy base collection process, when $|\B| = t$ to show (iii). If the estimators value is smaller than $1$, \Cref{eq_upperbound} holds for all elements in $A$, since at this point,  \eqref{eq_estimator} becomes
 \begin{equation*} \sum_{e \in A}(1+\epsilon)^{L^{\B}(e)-d t-\epsilon t/\gamma}.
\end{equation*}
Assume there was an element $e \in A$ for which $\ell^{\B}(e) > d + \eps/\gamma$, then the value of the estimator, \Cref{eq_estimator}, would also be at least $1$.

Finally, we need to consider the greedy addition of a new base to the collection and show that it cannot increase the value of the estimator to get (ii). We consider the step that adds the greedy base $B$ to $\B$. The new value of the estimator becomes 
 \begin{equation*} \sum_{e \in A}\frac{(1+\epsilon)^{L^{\B \cup \{B\}}(e)}(1+\epsilon d)^{t-|\B|-1}}{(1+\epsilon)^{d t+\epsilon t/\gamma}}.
\end{equation*}
The difference between the value before and after the deletion is then given by
\begin{equation*} \sum_{e \in A}\frac{((1+\epsilon)^{L^\B(e)}(1+\eps d) - (1+\epsilon)^{L^{\B \cup \{B\}}(e)})(1+\epsilon d)^{t-|\B|-1}}{(1+\epsilon)^{d t+\epsilon t/\gamma}}.
\end{equation*}
Hence, it suffices to show the following upper bound for the quantity
\begin{equation*} q(B) := \sum_{e \in A}(1+\epsilon)^{L^{\B \cup \{B\}}(e)} \leq \sum_{e \in A}(1+\epsilon)^{L^{\B}(e)}(1+\epsilon d).
\end{equation*}
If $B$ was a base randomly picked from $\Pi_\M$ the expected value of the quantity would be
\begin{align*}  \mathbb{E}_{B \in \Pi_\M}q(B) &= \sum_{e \in A}(1+\epsilon)^{L^{\B}(e)}((1+\epsilon)\Pr_{B\in\Pi_\M}(e\in B) + 1 - \Pr_{B\in\Pi_\M}(e\in B))\\
&\leq \sum_{e \in A}(1+\epsilon)^{L^{\B}(e)}(1 + \epsilon d).
\end{align*}
The inequality follows from \Cref{lem_prob} because every element $e \in A$ has ideal relative load $\ell^*(e) \leq d$. The next step is to show that the greedy base has a smaller quantity $q$ than for any base in $\Pi_\M$. We also have
\begin{equation*} q(B) = \sum_{e \in A}(1+\epsilon)^{L^{\B}(e)} + \epsilon \sum_{e \in B \cap A}(1+\epsilon)^{L^{\B}(e)}.
\end{equation*}
Hence going forward, it suffices to compare the values of the latter sum. We write 
\begin{equation*}
c(B) := \sum_{e \in B}(1+\epsilon)^{L^{\B}(e)}.
\end{equation*}
We need to show 
\begin{equation*} c(B \cap A) \leq \min_{B'\in \Pi_\M} c(B' \cap A).
\end{equation*}

Now, we want to find a set $S$, such that $c(S)$ is an upper bound for $c(B \cap A)$ as well as a lower bound for $\min_{B'\in \Pi_\M} c(B' \cap A)$. A suitable set is given by a minimum weight base $B_A$ for the restricted matroid $\M|A := (A, \mathcal{I}\;|\;A)$, where $\mathcal{I}|A := \{ X \in \mathcal{I} \; | \; X \subseteq A\}$.
\begin{claim} \label{claim_restricted_base}
   For any minimum weight base $B$ of $\M$ with respect to the loads $L^{\B}( \cdot)$ and a set $A \subseteq E$ there is a minimum weight base $B_A$ with respect to the loads $L^{\B}(\cdot)$ for $\M|A$, such that $B \cap A \subseteq B_A$.
\end{claim}
\begin{proof}
Given a minimum weight base $B$ for $\M$ we show how to construct a suitable base $B_A$ using the greedy algorithm for computing a minimum weight base. %Let $<_B$ be a strict total order on the elements of $B$, where according to the loads $L^{\B}( \cdot)$, tie-breaking in an arbitrary way (e.g. by comparing original labels). 
First we relabel the elements of $E$ to be increasing according to the following order: for $a,b \in E$ we have $e <_B f$ if and only if
\begin{itemize}
    \item $L^{\B}(e) < L^{\B}(f)$; or
    \item $L^{\B}(e) = L^{\B}(f)$ and
    \begin{itemize}
        \item $e \in B \land f \not \in B$; or otherwise
        \item $e < f$ lexicographically.
    \end{itemize}
\end{itemize}
Note that the relabeling of the elements in this way results in a non-decreasing order with respect to $L^{\B}( \cdot)$ regardless of the choice of $B$.
Now we build sets $S_E$ and $S_A$ using the greedy algorithm that iterates over the elements in $E$ according to the order given by the new labels. At the end of the algorithm $S_E$ will be the base $B$, and $S_A$ will be a minimum weight base of $\M|A$, such that $B \cap A \subseteq S_A$. Initially both sets are empty. In iteration $i$, when $e_i \in E$ is considered, do:
\begin{itemize}
    \item if $S_E + e_i \in \mathcal{I}$, add $e_i$ to $S_E$;
    \item if $e_i \in A$ and $S_A + e_i \in \mathcal{I}$, add $e_i$ to $S_A$.
\end{itemize}
The elements of both $E$ and $A$ are processed in non-decreasing order, simultaneously computing a minimum weight base for $\M$ as well as $\M|A$ with respect to $L^{\B}( \cdot)$.
First, we show that at any point of the algorithm $S_E$ spans $S_A$. Initially, this is true as both sets are empty. Consider iteration $i$ processing element $e_i$. If $S_A$ is not altered the claim stays true. If $e_i$ is added to $S_A$ then $e_i$ is either added to $S_E$ as well, or $S_E + e_i$ was not in $\mathcal{I}$. In both cases $e_i$ is spanned by $S_E$ which proves the claim. Further, we get that $\Span(S_A) \subseteq \Span(S_E)$ and hence we have for all $e_i \in E$ that $S_E + e_i \in \mathcal{I}$ implies $S_A + e_i \in \mathcal{I}$. Therefore all $e_i \in A$ that are added to $S_E$ are added to $S_A$ as well.

Now it remains to prove that after all elements have been processed indeed $S_E = B$. Again, the proof is by induction. Let $B_i := \{e \in B | e <_B e_i\} + e_i$ and show that after processing iteration $i$ we have $S_E = B_i$. In the case that $e_i \in B$, we have that $S_E + e_i \in \mathcal{I}$ holds and $e_i$ is added to $S_E$. Now consider the case that $e_i \not \in B$. If $S_E + e_i \not \in \mathcal{I}$ then $e_i$ is not added to $S_E$ and the claim holds. $S_E + e_i \in \mathcal{I}$ cannot hold as otherwise we could create a base $B'$ by adding elements from $B$ to $S_E + e_i$ according to the exchange property. We would get $B' = (B - b) + e_i$ for some element $b \in B \setminus S_E$ to be processed in a future iteration. Hence, $e_i <_B b$ and since $e_i \not \in b$ and $b \in B$ that is only possible if $L^{\B}(e_i) < L^{\B}(b)$ and $B$ would not have been a minimum weight base. 

Putting both properties of the algorithm together, after processing the last element, we get that $B \cap A \subseteq S_A$ and $S_A$ is a minimum weight base for $\M|A$.
\end{proof}

Now, we get the following: \begin{equation*} c(B \cap A) \leq c(B_A) \leq \min_{B'\in \Pi_\M} c(B' \cap A),
\end{equation*}
where the first inequality holds due to \Cref{claim_restricted_base} and for the second inequality we note that, for any $B' \in \Pi_\M$, we have that $B' \cap A$ is a base for $\M|A$. Otherwise $B' \cap A$ would not span $A$ and there would be an element $e \in A \setminus B'$ such that $(B' \cap A) +e \in \I$.  By the exchange property we could add elements from $B$ to $(B' \cap A) + e$ until we have another base for $\M$. Since the new base contains all elements from $B$ except for one element from $e' \in E \setminus A$ with $\ell^*(e') > d$ that is replaced with $e$ with $\ell^*(e) \leq d$ this contradicts that $B'$ is a minimum weight base with respect to the ideal relative loads $\ell^*(\cdot)$. Therefore, for any base $B''$ of $\M|A$, specifically also for $B' \cap A$, we have $c(B_A) \leq c(B'' \cap A)$ as $B_A$ is also a minimum weight base with respect to $(1 + \epsilon)^{L^{\B}(\cdot)}$ since the ordering between the elements stays the same if this value is considered instead of $L^{\B}(\cdot)$.
\end{proof}

Next, we show the corresponding lower bound following the same structure as above.
\begin{claim}
    For all elements $e \in E$ with $\ell^*(e) \le 1/\gamma$ we have
    \begin{equation*}\label{eq_lowerbound}
        \ell^{\B}(e) \ge \ell^*(e) - \eps/\gamma.
    \end{equation*}
\end{claim}

\begin{proof}
 Again, we consider $d$ with $0 \leq d \leq 1/\gamma$. Now, let $A := \{e \in E | \ell^*(e) \geq d \land \ell^*(e) \leq 1/\gamma\}$. As in \cite{Thorup07} we consider the new estimator   \begin{equation}\label{eq_estimator2} \sum_{e \in A}\frac{(1-\epsilon)^{L^{\B}(e)}(1-\epsilon d)^{t-|\B|}}{(1-\epsilon)^{d t-\epsilon t/\gamma}}.
\end{equation} 
We now show analogs of (i), (ii) and (iii) from the previous proof.
As all steps to prove that (i) the previous estimator initially has a value below 1 do not depend on whether $\epsilon$ is negative or positive, it directly translates to the new estimator. Similarly, the proof that (iii) the value of the estimator being smaller than 1 at the end of the process implies  $\ell^{\B}(e) \leq \ell^*(e) - \eps/\gamma$ is also analogous. Hence, we go on to show (ii) that adding a greedy base does not increase the new estimators value. Again, we define the quantity $q(B)$ to match the estimator. As above, we need to show that 
\begin{equation*} q(B) := \sum_{e \in A}(1-\epsilon)^{L^{\B \cup \{B\}}(e)} \leq \sum_{e \in A}(1-\epsilon)^{L^{\B}(e)}(1-\epsilon d).
\end{equation*}
Since now $\Pr_{B\in\Pi_\M}(e\in B) \geq d$ for every $e \in A$ we get 
\begin{equation*}  \mathbb{E}_{B \in \Pi_\M}q(B) \leq \sum_{e \in A}(1-\epsilon)^{L^{\B}(e)}(1 - \epsilon d).
\end{equation*}
Thus again, we need to show that for a greedy base $B$ we have $q(B) < q(B')$ for any $B' \in \Pi_\M$. As in the proof for the upper bound 
\begin{equation*} q(B) := \sum_{e \in A}(1-\epsilon)^{L^{\B}(e)} - \epsilon \sum_{B \cap A}(1-\epsilon)^{L^{\B}(e)}
\end{equation*}
and, as above, it suffices to compare the values $c(S) := -\sum_{e\in S}(1-\epsilon)^{L^{\B}(e)}$. Again, we want to find a set that simultaneously gives an upper bound for $c(B \cap A)$ and a lower bound for $\min_{B'\in \Pi_\M} c(B' \cap A)$. We will show that such a set is given by a minimum base of the contracted matroid $\M \cdot A$. For equivalent definitions and more background on contracted matroids, see, e.g., \cite{oxley92}. 
\begin{definition}
    Let $B_{\overline A}$ be a base of $\M|\overline A$ and let $\mathcal{I} \cdot A := \{X \subseteq A \;|\; X \cup B_{\overline A} \in \mathcal{I}\}$. The matroid $\M \cdot A := (A, \mathcal{I} \cdot A)$ is the contraction of $\M$ to $A$.
\end{definition}

For completeness, let us show this is well-defined. 

\begin{claim}
    Let $B$ and $B'$ be bases of $\M|\overline A$, then \[\{X \subseteq A \;|\; X \cup B \in \mathcal{I}\} = \{X \subseteq A \;|\; X \cup B' \in \mathcal{I}\}.\]
\end{claim}
\begin{proof}
    Let $X\subseteq A$ s.t.\ $X\cup B\in \I$, then we have to show that $X \cup B'\in \I$. Suppose not: $X\cup B'\notin \I$, but since $B$ is a base for $\M|\overline A$, we have that $B'\subseteq \Span(B)$, so then $X\cup B \notin \I$, a contradiction. Since this proof is symmetric in $B$ and $B'$, it shows equality. 
\end{proof}

When we consider $\M \cdot A$ for the graphic matroid, then it corresponds to contracting the elements in $\overline A$, usually denoted by $M/\overline{A}$. Since there is no such notion of vertices in matroids, we cannot contract two (or more) vertices into one. We choose to use the notation $M\cdot A$ to stress this difference. Moreover, we only need the contracted matroids conceptually for our analysis, which means we do not need to specify how to query them.

Now, we show the following claim in order to upper bound $c(B \cap A)$.
\begin{claim} \label{claim_contracted_base}
   For any minimum weight base $B$ of $\M$ with respect to the loads $L^{\B}( \cdot)$ and a set $A \subseteq E$ there is a minimum weight base $B_{A}$ with respect to the loads $L^{\B}(\cdot)$ for $\M \cdot A$, such that $B_{A} \subseteq B \cap A$.
\end{claim}
\begin{proof}
Again, we show this by altering the greedy algorithm to find a minimum weight base. We relabel the elements according to the order as described previously. We build the set $S_E$ exactly as before. Simultaneously we build $S_{ A}$ as follows. Let $B_{\overline A}$ be a base for $\M|\overline A$. In iteration $i$, when $e_i$ is considered, add $e_i$ to  $S_{A}$ if $S_{A} \cup B_{\overline A} +e_i \in \mathcal{I}$. Note that $S_{A} \cup B_{\overline A} +e_i \in \mathcal{I}$ holds if and only if $S_{A} + e_i \in \mathcal{I} \cdot A$.
Hence, the algorithm computes a minimum weight base for $\M \cdot A$ in addition to computing the minimum weight base $B$. Now we show that at any point during the algorithm $S_{A} \cup B_{\overline A}$ spans~$S_E$. Initially, this is trivially true. Consider an iteration $i$ in which $e_i$ gets added to $S_E$ otherwise the property continues to hold. If $S_{A} \cup B_{\overline A} + e_i \in \mathcal{I}$, $e_i$ is added to $S_{A}$ as well and $S_{A}$ spans $S_E$ again. If $S_{A} \cup B_{\overline A} + e_i \not \in \mathcal{I}$ then  $S_{A} \cup B_{\overline A}$ spans $S_E + e_i$ by induction. Again, this immediately gives $\Span(S_E) \subseteq \Span(S_{A} \cup B_{\overline A})$ and any element that gets added to $S_{A}$ is also added to $S_E$ during the algorithm. As $B_{\overline A}$ spans $\overline A$, all elements in $S_{A}$ are elements of $A$. Hence, $B \cap A$ contains a minimum weight base $B_A$ for $\M \cdot A$. \end{proof}
As a result, for a base $B_{A}$ as described in \cref{claim_contracted_base} we have
\begin{equation*} c(B \cap A) \leq c(B_A).
\end{equation*}
It also holds that $B' \cap A$ is a base for $\M \cdot A$ for any $B' \in \Pi_\M$. Assume that this was not the case and consider a base of $B_{\overline A}$ of $\M|\overline A$. Then $(B' \cap A) \cup B_{\overline A}$ is not a base of $\M$. Then we could build another base of $\M$ by the exchange property by adding elements from $B'$ to $B_{\overline A}$. We would only add elements of $B \cap A$ in this way as $B_{\overline A}$ spans all other elements of $B'$, but there would be at least one element of $B' \cap A$ that we cannot add. This element is replaced by an element of $B_{\overline A} \subseteq \overline A$. Since $\ell^*(e) < \ell^*(e')$ for any $e \in \overline A$ and $e' \in A$, $B'$ would not have been a minimum weight base for $\M$. Again, because it does not change the order between the elements, the minimum weight bases of $\M \cdot A$ stay the same whether we consider $-(1-\epsilon)^{L^\B(\cdot)}$ or $L^\B(\cdot)$, hence we further get 
\begin{equation*} c(B \cap A) \leq c(B_A) \leq \min_{B' \in \Pi_\M} c(B' \cap A).
\end{equation*}

Now, it remains to show that \Cref{eq:greedy_coll_high} holds for all elements $e \in E$ with ${\ell^*(e) \geq 1/\gamma}$. We do so again in two parts using an estimator for the number of violations. Let ${1/\gamma \leq d \leq 1/\Phi}$. First, we show that $\ell^\B(e) \leq (1 + \eps)\ell$. To do this, let $A := \{\ell^*(e) \leq d\}$ and consider the estimator
\begin{equation*}
\sum_{e \in A}\frac{(1+\epsilon)^{L^{\B}(e)}(1+\epsilon d)^{t-|\B|}}{(1+\epsilon)^{(1+\eps)d t}}.    
\end{equation*}
Note that adding a greedy base still does not increase the value of the estimator (ii), as only the denominator changed compared to the previous proof for \Cref{eq:greedy_coll_low}. Now, we consider the end of the process, when all bases have been added greedily. Again, we need to show that (iii) the estimator being smaller than $1$ implies $\ell^\B(e) \leq (1 + \eps) d$ for all $e \in A$. when $t = |\B|$ the new estimator becomes
\begin{equation*}
\sum_{e \in A}(1+\epsilon)^{L^{\B}(e) - (1 + \eps)d t}.    
\end{equation*}
Hence, if there is an $e \in A$ with $\ell^\B(e) \geq (1 + \eps) d$ then the above summand corresponding to $e$ would be at least $1$.

Now, we need to show that (i) initially the value of this new estimator was also below $1$:
\begin{equation*}
\sum_{e' \in A}\frac{(1+\epsilon d)^{t}}{(1+\epsilon)^{(1+\eps)d t}} \leq n\left(\frac{e^{\epsilon}}{(1+\epsilon)^{(1+\epsilon)}}\right)^{d t} \leq n\left(\frac{e^{\epsilon}}{(1+\epsilon)^{(1+\epsilon)}}\right)^{1/\gamma t} \leq 1.
\end{equation*}
The second to last inequality follows since $d \geq 1/\gamma$ and $e^{\epsilon}/(1+\epsilon)^{(1+\epsilon)} \leq 1$ and the last inequality is as in the first part of the proof.

It remains to show that $\ell^\B(e) \leq (1 - \eps)d$ holds for all $e \in E$ with $d \leq \ell^*(e)$. 

All steps are analogous to the $(1+ \eps)$-case, with $A := \{\ell^*(e) \geq d\}$ and the estimator 
\begin{equation*}
\sum_{e \in A}\frac{(1-\epsilon)^{L^{\B}(e)}(1-\epsilon d)^{t-|\B|}}{(1-\epsilon)^{(1-\eps)d t}}.    
\end{equation*}
 This concludes the proof of \Cref{lem:greedy_collnew}.
\end{proof}

\printbibliography[heading=bibintoc] % Make bibliography show up in table of contents

% \newpage 
% \appendix

\end{document}